\documentclass[a4paper]{article}

\usepackage{fullpage}
\usepackage{amsthm, amsmath}
\usepackage{url}
\usepackage{times}
\usepackage[utf8]{inputenc}
\usepackage{pgfplots}
\pgfplotsset{compat=1.14}
\usepackage{xspace}
\usepackage{mathtools}
\usepackage{cleveref}
\usepackage{paralist}
\usepackage{algorithm} 
\usepackage[noend]{algorithmic}

\newtheorem{definition}{Definition}
\newtheorem{theorem}{Theorem}
\newtheorem{lemma}{Lemma}

\newtheorem{claim}{Claim}

\newcommand{\E}{\ensuremath{E}} 
\newcommand{\poly}{\ensuremath{\mathsf{poly}}}
\newcommand{\prob}{\text{Pr}\xspace}

\newcommand{\cA}{\mathcal{A}\xspace}

\title{Breaking the $\tilde\Omega(\sqrt{n})$ Barrier:\\
   Fast Consensus under a Late Adversary\thanks{This work was supported by the LMS Computer Science Scheme 7 Grant Ref No: SC7-1516-11 and by the German Research Foundation (DFG) within the Collaborative Research Center ``On-The-Fly Computing'' (SFB 901).}}

\author{
Peter Robinson\thanks{McMaster University, Hamilton, Canada. \hbox{E-mail}:~{\tt peter.robinson@mcmaster.ca}}
\and
Christian Scheideler\thanks{Paderborn University, Germany.
\hbox{E-mail}:~{\tt scheideler@mail.upb.de}}
\and
Alexander Setzer\thanks{Paderborn University, Germany.
\hbox{E-mail}:~{\tt asetzer@mail.upb.de}}
}

\begin{document}

\maketitle

\begin{abstract}
We study the consensus problem in a synchronous distributed system of $n$ nodes under an adaptive adversary that has a slightly outdated view of the system and can block all incoming and outgoing communication of a constant fraction of the nodes in each round.
Motivated by a result of Ben-Or and Bar-Joseph (1998), showing that any consensus algorithm that is resilient against a linear number of crash faults requires $\tilde \Omega(\sqrt{n})$ rounds in an $n$-node network against an \emph{adaptive} adversary, we consider a late adaptive adversary, who has full knowledge of the network state at the beginning of the \emph{previous} round and unlimited computational power, but is oblivious to the current state of the nodes.
Our main contributions are randomized distributed algorithms that achieve consensus with high probability among all except a small constant fraction of the nodes (i.e.,\ ``almost-everywhere'') against a late adaptive adversary who can block up to $\epsilon n$ nodes in each round, for a small constant $\epsilon >0$.
Our first protocol achieves binary almost-everywhere consensus and also guarantees a decision on the majority input value, thus ensuring plurality consensus.
We also present an algorithm that achieves the same time complexity for multi-value consensus.
Both of our algorithms succeed in $O(\log n)$ rounds with high probability, thus showing an exponential gap to the $\tilde\Omega(\sqrt{n})$ lower bound of Ben-Or and Bar-Joseph for strongly adaptive crash-failure adversaries, which can be strengthened to $\Omega(n)$ when allowing the adversary to block nodes instead of permanently crashing them.
Our algorithms are scalable to large systems as each node contacts only an (amortized) constant number of peers in each communication round.
We show that our algorithms are optimal up to constant (resp.\ sub-logarithmic) factors by proving that every almost-everywhere consensus protocol takes $\Omega(\log_d n)$ rounds in the worst case, where $d$ is an upper bound on the number of communication requests initiated per node in each round.
We complement our theoretical results with an experimental evaluation of the binary almost-everywhere consensus protocol revealing a short convergence time even against an adversary blocking a large fraction of nodes.
\end{abstract}

\section{Introduction} \label{sec:intro}

Reaching consensus among the members of a distributed system is a fundamental ingredient for building a reliable distributed system out of unreliable parts (c.f., \cite{Lyn96,attiyawelch}), with applications ranging from synchronization of processes to distributed commit in databases.

Given that several impossibility results in asynchronous systems \cite{FLP85}  and pessimistic lower bounds \cite{dolevstrong,FLM86} for synchronous systems have been discovered for deterministic algorithms when nodes can be subjected to failures, much of the recent work on the consensus problem focuses on using randomization to defeat the adversary (see \cite{aspnes} for a comprehensive overview).
Early randomized consensus algorithms include the work of Ben-Or \cite{benor}, which achieves consensus even in the asynchronous setting, where nodes perform steps at different speeds and messages can exhibit delays.
For the more severe case of Byzantine failures, where some nodes are directly controlled by an adversary and hence deviate from the protocol, Rabin \cite{rabin} showed that it is possible to achieve agreement in a constant number of rounds in expectation when nodes have access to a shared common coin.
The recent work of \cite{DBLP:journals/jacm/KingS11}, which assumes private communication channels between non-corrupted nodes, and \cite{DBLP:journals/jacm/MostefaouiMR15} achieved a significant reduction of the required message complexity of randomized Byzantine agreement protocols.

Since all of the classic agreement protocols mentioned above require all-to-all communication and hence do not scale well to large systems, one direction of research is to study the consensus problems in systems where the nodes are interconnected by a sparse communication network.
Considering that the adversary can isolate (small) regions of a sparse network, achieving agreement among \emph{all} nodes is impossible without strong guarantees on the network's connectivity \cite{FLM86}.
However, for many practical applications such as achieving eventual consistency in a peer-to-peer network (e.g.,\ \cite{nakamoto2008bitcoin}), it is sufficient if a common decision is reached ``almost everywhere'', i.e., among all but a small fraction of nodes.
The \emph{almost-everywhere consensus} problem was first introduced in \cite{dwork1988fault} (called almost-everywhere agreement) as a way to circumvent the classical impossibility results for  consensus and Byzantine agreement in sparse communication networks and has been studied under various failure models, such as
strongly adaptive and oblivious adversaries.
More specifically, \cite{dwork1988fault} provided an almost-everywhere agreement algorithm that can tolerate $O(n/\log n)$ Byzantine faults in a sparse network. Later, \cite{upfal1992tolerating} improved over these results by presenting a protocol that tolerates a linear number of failures. However, both of these protocols require nodes to use a polynomial number of bits of communication.
More recently, almost-everywhere consensus with Byzantine nodes was studied in the context of bounded-degree peer-to-peer networks in the work of \cite{king2006towards}.
We point out that \cite{dwork1988fault,upfal1992tolerating,king2006towards} require nodes to be equipped with a priori knowledge about the network topology that needs to be hard-coded into the distributed algorithm.

Other works consider almost-everywhere consensus in dynamically changing expander network topologies~\cite{podc13,soda12} where the nodes can be subjected to churn over time.
More specifically, \cite{podc13} provide an algorithm using logarithmic-size messages in the presence of $\tilde O(\sqrt{n})$\footnote{The notation $\tilde O(.)$ hides a polylogarithmic factor.} (obliviously-controlled) churn and Byzantine nodes, while an extension of the algorithm can tolerate a strongly adaptive adversary albeit at the cost of requiring polynomial message sizes.
When there are no Byzantine nodes, \cite{soda12} show that it is possible to tolerate even $O(n)$-churn when the adversary is oblivious.
However, none of these algorithms work against a late adversary (defined in \Cref{sec:model}) who is adaptive \emph{and} can fail a constant fraction of nodes.

\subsection{Our Main Results} \label{sec:mainresults}

In this work, we consider the almost-everywhere consensus problem under a ``late'' adaptive adversary who has a slightly outdated view of the entire system when deciding its next move. Studying late adversaries is reasonable as it is unrealistic to assume that an adversary can make instantaneous decisions based on the global system state, and our results demonstrate that it already makes a huge difference if the adversary is just late by a single communication round.

We present new distributed algorithms for almost-everywhere consensus in a system of $n$ synchronous, completely interconnected nodes that can withstand attacks on a constant fraction of the nodes in each round by a late adversary who has full knowledge of the system state at the start of the previous round.
In our protocols, every node sends and receives at most $O(\log n)$ messages per round with high probability\footnote{We say an event $E$ holds with high probability (w.h.p.) if $\prob[E] \geq 1 - n^{-\Omega(1)}$.}.
In particular we present the following results under a late adaptive adversary that can block the communication of up to $\epsilon n$ nodes in each round, for a constant $\epsilon>0$:
\begin{enumerate}
  \item[(1)] In \Cref{sec:binary}, we provide an $O(\log n)$-rounds algorithm that achieves almost-everywhere binary consensus with high probability, where each node contacts $O(1)$ other nodes per round. Moreover, the algorithm achieves plurality consensus (w.h.p.).
  \item[(2)] In \Cref{sec:multivalue}, we present an $O(\log n)$-rounds algorithm that solves almost-everywhere multi-value consensus (w.h.p.). (Instead of the plurality guarantee, this algorithm satisfies the classical validity condition of consensus; see \Cref{sec:model}.)
  \item[(3)] In \Cref{sec:lowerbound} we show a lower bound on the time complexity of any algorithm that achieves almost-everywhere agreement with high probability assuming that each node can communicate with at most $d$ nodes (chosen by the algorithm and not by the adversary), requires $\Omega(\log_d n )$ rounds in the worst case.
    This implies that our algorithm for binary consensus (cf.\ \Cref{sec:binary}) is asymptotically optimal, whereas the multi-value algorithm (cf.\ \Cref{sec:multivalue}) is optimal up to sub-logarithmic factors, due to requiring $d = \Theta(\log n)$.
  \item[(4)] In \Cref{sec:aelowerbound} we deal with another lower bound: Note that the work of \cite{barjoseph} shows that if a \emph{strongly adaptive} adversary who observes the current system state including the output of the random coinflips can fail up to $t$ of the $n$ nodes in each round, then with high probability the adversary can force any consensus algorithm to run for $\Omega(t/\sqrt{n \log n})$ rounds.
  This result holds even without restricting the communication of the nodes. 
  In fact, when considering a strongly adaptive adversary who can block a set of nodes, instead of permanently crashing them, it is possible to show that $\Omega(n)$ rounds are necessary, which we elaborate on in \Cref{thm:lb_fullyadaptive}.
  \item[(5)] Finally, in \Cref{sec:experiments} we present an experimental evaluation of the algorithm from \Cref{sec:binary} with different parameters, suggesting a good practical performance of at most $2 \log(n)$ rounds convergence time (on average) of that algorithm and variants of it, even against an adversary who can block a large number of nodes.
\end{enumerate}

\subsection{Technical Contribution}\label{sec:tc}

The main challenge that we have to overcome is the fact that the adversary has an almost up-to-date knowledge of the state of the system.
The key ingredient in our approach in handling a linear number of adversarial attacks is to limit the extent to which the previous state of a node influences the computations of the algorithms.

For the binary consensus problem, at first glance the median rule presented in \cite{doerr2011adversary} seems to be a good candidate to handle a late adversary. 
It uses a simple ``pull'' strategy in which each node requests the values from $2$ nodes chosen uniformly and independently at random, and selects the median of these two values and its own value as its new value. 
However, ``pull'' strategies, which inherently require two communication rounds to compute a new value, do not work against a late adaptive adversary, as a late adversary can anticipate the values received at the end of the response round by the end of the request round.
This essentially makes the adversary as powerful as the strongly adaptive one.
Another problem is that the value of a node depends on its previous value.
In fact, it can be shown that the median rule fails w.h.p. against a late adversary that can just cause $\Omega(\sqrt{n \log n})$ faults.
Thus, we converted the median rule into a ``push'' strategy, in which each node spreads its value in every communication round to other nodes selected uniformly at random, and picks the majority of a random subset of received values for its new value.
When using such a procedure, however, the probabilities of the nodes to receive certain values are no longer independent from each other. 
This requires a much more involved analysis, as we can no longer rely on standard Chernoff bounds and the central limit theorem for showing a drift towards one value and need to use the Paley-Zygmund inequality instead. 
Also, it can happen that nodes do not receive enough values to compute a new value. 
In that case, we need an additional undecided state for the nodes to ensure independency from the past node state, which has the effect that the number of nodes with a value may vary from round to round. 
This, in turn, causes the support of the dominating value to no longer increase in absolute terms but rather only relatively to the number of nodes that are not in an undecided state.

While the algorithm for binary almost-everywhere consensus (\Cref{sec:binary}) solely bases its computations on the values received from other nodes, in our algorithm for handling multi-value inputs (\Cref{sec:multivalue}) a node does keep its value unless it is blocked or learns about a larger value. Hence, eventually a consensus is reached on the maximum value in the system. However, as the events of nodes joining the set that is informed about the maximum value are not independent, we use the method of bounded differences and Azuma's Inequality that a consensus is reached even under a late adversary that can block a constant fraction of the nodes.

\subsection{Other Related Work} \label{sec:otherrelated}
As mentioned above, our algorithm for binary consensus is inspired by \cite{doerr2011adversary} which considers almost-stable consensus under an adversary that may adaptively change the values of up to $\sqrt{n}$ nodes after each round.
The authors derive an $O(\log n + \log k \log\log n)$ runtime bound w.h.p. for an arbitrary number of distinct initial values $k$ and $O(\log n)$ if $k=2$.
The latter bound is tight (even without an adversary) as we show in this work.
To achieve their results, the authors use the median rule. In the binary case this is equivalent to taking the majority of the received values, which resembles the 3-majority rule.
\cite{Becchetti2016stabilizing} study the 3-majority rule under an $o(\sqrt{n})$-adaptive adversary, showing that it converges in time polynomial in $k$ and $n$ w.h.p.

Our algorithm in \Cref{sec:binary} guarantees a decision on the majority input value (if any), which was studied in non-adversarial settings by \cite{BCNPST}. They present an algorithm with runtime $\Theta(k\log n)$ w.h.p. for $k$ being the number of distinct initial values, which they prove to be tight for a wide range of values of $k$ and the initial bias.
For certain values of $k$, the runtime of this algorithm is outperformed by \cite{BCNPS} in which each node may contact one neighbor at random in each round.
The authors of the latter paper consider what they call the Undecided-State Dynamics (originally introduced in \cite{AAE} as Third-State dynamics), which makes use of an additional ``undecided'' value different from all input values.
This idea has also been used in \cite{berenbrink_et_al:LIPIcs:2016:6271} and \cite{Ghaffari:2016:PGA:2933057.2933097} to solve plurality consensus with small memory overhead and an exponential improvement on the runtime compared to \cite{BCNPST}.
In contrast to these works that do not consider adversaries, our algorithms, which also use undecided values, do so to prevent the adversary from influencing the distribution of values by releasing previously blocked nodes.
Note that all consensus protocols mentioned thus far use a ``pull'' strategy, under which a late adversary is as powerful as a fully adaptive one (as was argued in Section~\ref{sec:tc}).
Thus none of the analyses of these protocols could be trivially adapted to achieve a consensus under a constant fraction of adversarial nodes in $\tilde o(\sqrt{n})$ rounds .

The idea of considering a late adversary who knows a previous state of the system has only been studied in a few papers so far (e.g., \cite{DBLP:conf/wdag/AwerbuchS07,kuhn-opodis15, DBLP:conf/algosensors/AhmadiK16, drees16lateadversary, DBLP:journals/corr/KlonowskiKM17}).
As one of the first attempts, \cite{DBLP:conf/wdag/AwerbuchS07} studies how to make an information system resilient against a blocking adversary that has knowledge of the system up to some point in time.
In \cite{kuhn-opodis15} and \cite{DBLP:conf/algosensors/AhmadiK16}, an adversary is called $\tau$-oblivious if it knows the randomness of the nodes up to round $r-\tau$ to make a decision for round $r$, in which the special case of a $0$-oblivious adversary is called a strongly adaptive adversary.
The authors of \cite{drees16lateadversary} consider a so-called $t$-late adversary that bases its decision in round $r$ on knowledge of the complete state of nodes in round $r-t$ and use this concept to maintain connectivity under an $O(\log \log n)$-late adversary who may cause churn and block nodes.
\cite{DBLP:journals/corr/KlonowskiKM17} introduces the notion of a $c$-round-delayed adversary which resembles a $c$-oblivious adversary in the notation of \cite{kuhn-opodis15}.

A failure model related to the blocking adversary is the crash-recovery model considered for asynchronous consensus in \cite{aguilera2000failure}. We point out that, in contrast to our model, in \cite{aguilera2000failure} nodes lose their entire state upon recovery and can only recover state that was written to a stable storage.

\cite{DBLP:conf/soda/GilbertK10} presents a protocol that achieves consensus against an oblivious adversary requiring only $O(n)$ message complexity, which is clearly optimal. While the overall communication of their algorithm is small, it requires an asymmetric communication load, as some nodes need to send $\Omega(\sqrt{n})$ messages.
Other work that considers consensus in models with restricted communication includes \cite{newport}, who studies agreement in the radio network model of wireless communication.

\subsection{Computing Model and Problem Definition} \label{sec:model}

We consider a system of $n$ anonymous, completely interconnected nodes that operate in synchronous rounds. 
In each round, every node first receives all messages from the previous round (that are not blocked by the adversary), then does some arbitrary finite internal computation, and then sends off messages to other nodes. 

\noindent\textbf{Adversarial Model.}
Due to the lower bound results of \cite{barjoseph,podc13} discussed in \Cref{sec:mainresults}, achieving agreement in logarithmic time is out of reach against a strongly adaptive adversary who can block a linear number of nodes per round, even when the amount of communication that a node can perform in each round is unrestricted.
This motivates us to consider a \emph{late $\epsilon$-bounded adaptive adversary}, who, at the start of each round $t\ge 2$, can observe the state of the entire system at the beginning of round $t-1$. Based on that state, the adversary can block any set $B_t$ of $\le\epsilon n$ nodes, with the effect that none of these nodes can receive or send messages in that round.
We point out that, in the very first round, the adversary knows the initial state of the system. However, it does not know the coin flips performed by nodes in this round.
Moreover, we assume that nodes know if they are currently blocked.
Notice that the adversary is still adaptive in the sense that it can tailor its next attack on the previously observed network state.

\begin{definition}
  Let $\gamma<\frac{1}{2}$ be a positive constant.
We assume that each node $u$ starts with an \emph{input value} chosen by the adversary from a domain of size polynomial in $n$ and, eventually, $u$ irrevocably decides on a value.
  We say that an algorithm solves \emph{almost-everywhere agreement in $t$ rounds with high probability and loss $\gamma n$}, if the algorithm terminates in $t$ rounds and the following properties hold w.h.p.:
\begin{compactdesc}
\item[(Almost-Everywhere Agreement)] All except at most $\gamma n$ nodes decide on the same value.
\item[(Validity)] If all nodes start with the same input value $x$, then $x$ is the only possible decision.
\end{compactdesc}
\end{definition}
In addition to the classic validity definition above, we are also interested in a stronger validity condition:\\
\noindent\textbf{(Plurality)}
If an almost-everywhere agreement algorithm $\cA$ guarantees the stronger validity condition that stipulates $x$ as the only possible decision if
the number of nodes having input $x$ exceeds the number of nodes having input $y$ for every other possible value $y \neq x$ by at least $s \geq \gamma n$. In this case, we say that \emph{$\cA$ solves almost-everywhere plurality consensus with initial bias $s$}.

\section{Binary Consensus} \label{sec:binary}
In this section, we assume that each node $u$ has a local variable $x_u$ that can hold the values 0, 1, or $\bot$ (i.e. ``undefined''). A node $u$ is called {\em undefined} if $x_u=\bot$ and otherwise {\em defined}. Consider the following protocol for some integers $k,\ell > 0$, where $k \geq \ell$ and $\ell$ odd:

\noindent{\bf $(k,\ell)$-majority algorithm.} In each round, every node $u$ acts according to the following rules:
\begin{compactitem}
\item \textbf{reset rule:} If $u$ has received less than $\ell$ values from the previous round, or it is blocked, then it sets $x_u:=\bot$ and skips the rest of this round.
\item \textbf{update rule:} Otherwise, $u$ picks $\ell$ of the received values uniformly at random, sets
 $x_u$ to the majority value of the $\ell$ values, and sends out $x_u$ to $k$ nodes chosen independently and uniformly at random.
\end{compactitem}
We assume here that if the adversary blocks $u$, then $u$ notices it. (In practice, it may simply not receive enough values, which would also have the effect in our protocol that $x_u$ is set to $\bot$.) In the rest of this section we show:
\begin{theorem} \label{thm:convergence}
For any initial assignment of values to the nodes with no undefined node and any $\epsilon \le 1/16$, $O(\log n)$ rounds of the $(6,3)$-majority rule suffice to ensure that all except a constant fraction of nodes hold the same value (w.h.p) against a late $\epsilon$-bounded blocking adversary.
\end{theorem}
On a high-level view, the proof works as follows:
Let $X_t$ be the number of nodes with value 0 at the end of round $t$, let $Y_t$ be the number of nodes with value 1 at the end of round $t$, and let $n_t := X_t + Y_t$ be the number of defined nodes at the end of round $t$.
First of all, observe that there will be at least $3n/4$ defined nodes at the end of round 1 since the adversary can block at most $n/16$ many nodes.
In Lemma~\ref{lem:mindefined} we show that this holds after every round.
In the remaining part, we assume without loss of generality that $X_t \le Y_t$ and define the {\em imbalance} at round $t$ by $\Delta_t = (Y_t-X_t)/2$ (which is non-negative by our assumption). 
Based on the imbalance $\Delta_t$ we distinguish between three cases. 
First, we show that if the imbalance $\Delta_t$ exceeds $n_t/4$, i.e., there is a value that at least $3/4$ of the defined nodes have, then this value will be spread to almost every node within $O(\log \log n)$ rounds w.h.p.
Second, we show that if $c\sqrt{n_t \ln n_t} \le \Delta_t < n_t/4$ for a sufficiently large constant $c$, the number of nodes having the dominating value will grow within $O(\log n)$ rounds w.h.p. such that it is held by at least $3/4$ of all defined nodes.
While the overall structure of the proof in the first two cases follows \cite{doerr2011adversary}, here we face additional complications due to dependencies among the random variables.
Note that these two cases directly imply that  the $(k,\ell)$-majority algorithm solves the plurality consensus problem for initial bias $s \geq \sqrt{n(1 + \epsilon) \log (n(1 + \epsilon))}$.
In the third case, we show that for very low imbalances, with high probability one of the two values will take over a sufficiently large number of the nodes such that the requirements of the second case are fulfilled.

To begin with the formal analysis, we introduce some results from probability theory.
\begin{definition}
Let $X_1,\ldots,X_n$ be binary random variables. We say that $X_1,\ldots,X_n$ are {\em $s$-wise independently upper bounded by} $p_1,\ldots,p_n$ if for any subset $S \subset \{1,2,\ldots,n\}$ with $|S| \le s$ and any $i \in \{1,2,\ldots,n\} \setminus S$ it holds that
$
  \Pr[X_i=1 \mid \bigwedge_{j \in S} X_j = 1] \le p_i.
$
\end{definition}
For these random variables the following variant of the Chernoff bounds can be shown:
\begin{lemma} \label{lem:chernoff}
Let $X_1,\ldots,X_n$ be binary random variables that are $k$-wise independently upper bounded by $p_1,\ldots,p_n$ and $X=\sum_{i=1}^n X_i$. Also, let $\mu = \sum_{i=1}^n p_i$. Then it holds for all $\delta > 0$ that
\begin{align*}
  \Pr[X \ge (1+\delta)\mu] & \le \left( \frac{e^{\delta}}{(1+\delta)^{1+\delta}}
  \right)^\mu
  \le e^{-\delta^2 \mu/(2(1+\delta/3))} \le e^{-\min\{\delta^2, \delta\} \mu/3}
\end{align*}
as long as $k \ge (1+\delta)\mu$.
\end{lemma}
This lemma directly follows from Theorem 3.52 in \cite{Sch00} which itself is based on results in \cite{SSS95}.

Certainly, there will be at least $3n/4$ defined nodes at the end of round 1 since the adversary can block at most $n/16$ many nodes. For all subsequent rounds it holds:

\begin{lemma}\label{lem:mindefined}
For any assignment of values to the nodes with at least $3n/4$ defined nodes at the end of round $t$, there will also be at least $3n/4$ defined nodes at the end of round $t+1$ w.h.p.
\end{lemma}
\begin{proof}
Recall that $n_t$ is the number of defined nodes at the end of round $t$. 
For any non-blocked node $u$ in round $t+1$ it holds that
$  \Pr[u \mbox{ receives no value}]  = (1-1/n)^{6n_t} \le e^{-6n_t/n}$ and
\begin{align*}
  \Pr[u \mbox{ receives one value}] & = \binom{6n_t}{1} (1/n) \cdot (1-1/n)^{6n_t-1}  \le \frac{6n_t}{n} \cdot e^{-(6n_t-1)/n},\\
  \Pr[u \mbox{ receives two values}] & = \binom{6n_t}{2} (1/n)^2 \cdot (1-1/n)^{6n_t-2} 
     \le \frac{1}{2} \left( \frac{6n_t}{n} \right)^2 \cdot e^{-(6n_t-2)/n}.
\end{align*}
Hence,
\[
  \Pr[u \mbox{ receives $\le 2$ values}] \le \left( 1+ \frac{6n_t}{n} + \frac{1}{2} \left( \frac{6n_t}{n} \right)^2 \right) \cdot e^{-(6n_t-2)/n}.
\]
If $n_t \ge 3n/4$ then $\Pr[u$ receives $\le 2$ values$] \le 2/11$ if $n$ is large enough. Since the adversary can block at most $n/16$ nodes, we get
$
  \E[ \mbox{number of undefined nodes} ] \le (2/11)(1-1/16)n + n/16 \le n/4.4.
$
Unfortunately, the probabilities of the nodes to be undefined are not independent, but they are $n/4$-wise independently upper bounded by $2/11$ as shown by the following claim.

\begin{claim} \label{cl:neg-cor}
For any subset $S$ of non-blocked nodes with $|S|\le n/4$ and any non-blocked node $u \not\in S$ it holds that
$
  \Pr[ u \mbox{ undefined} \mid \mbox{all nodes in $S$ undefined} ] \le 2/11.
$
\end{claim}
\begin{proof}
Let $s=|S|$. Suppose that exactly $k$ of the $6n_t$ values have been sent to nodes in $S$. Certainly, if all nodes in $S$ are undefined, then $k \le 2s$. It is easy to check that if $n$ is sufficiently large, then for any $n_t \ge 3n/4$ and any $k \le 2s$,
$(6n_t-k)/(n-s) \ge 6n_t/n$. Moreover, the functions $e^{-x}$, $x \cdot e^{-x}$ and $x^2 \cdot e^{-x}$ are monotonically decreasing for any $x \ge 2$. Hence, given that all nodes in $S$ are undefined, we get:
\begin{align*}
  \Pr[u & \mbox{ receives no value}]\\
   & = (1-1/(n-s))^{6n_t-k} \le e^{-(6n_t-k)/(n-s)} \le e^{-6n_t/n} \\
  \Pr[u & \mbox{ receives one value}] \\
   & = \binom{6n_t-k}{1} (1/(n-s)) \cdot (1-1/(n-s))^{6n_t-k-1} \\
   & \le \frac{6n_t-k}{n-s} \cdot e^{-(6n_t-k-1)/(n-s)} \le \frac{6n_t}{n} \cdot e^{-6n_t/n} \cdot e^{1/(n-s)} \\
   \Pr[u & \mbox{ receives two values}] \\
  & = \binom{6n_t-k}{2} (1/(n-s))^2 \cdot (1-1/(n-s))^{6n_t-k-2} \\
   & \le \frac{1}{2} \left( \frac{6n_t-k}{n-s} \right)^2 \cdot e^{-(6n_t-k-2)/(n-s)} \\
   & \le \frac{1}{2} \left( \frac{6n_t}{n} \right)^2 \cdot e^{-6n_t/n} \cdot e^{2/(n-s)}
\end{align*}
Hence, given that all nodes in $S$ are undefined,
\begin{align*}
  \Pr[u & \mbox{ receives $\le 2$ values}] \\
  & \le \left( 1+ \frac{6n_t}{n} + \frac{1}{2} \left( \frac{6n_t}{n} \right)^2 \right) \cdot e^{-6n_t/n} \cdot e^{2/(n-s)} \le 2/11
\end{align*}
as long as $|S| \le n/4$ and $n$ is sufficiently large.
\end{proof}

Therefore, we can make use of the Chernoff bounds in Lemma~\ref{lem:chernoff} to show that the number of undefined nodes is at most $n/4$ w.h.p.
This completes the proof of the lemma.
\end{proof}

In the remaining part of the analysis, we carry out the aforementioned case distinction based on the value of $\Delta_t$.

\subsubsection*{Case 1: {\boldmath $\Delta_t \ge n_t/4$}}
\newcommand{\lemcaseone}{
If there is a round $t_0$ with $\Delta_{t_0} \ge n_{t_0}/4$, then there is a round
$t_1=t_0+O(\log \log n)$ at which we reach an almost-everywhere consensus, w.h.p.
}
\begin{lemma} \label{lem:caseone}
\lemcaseone
\end{lemma}
\begin{proof}
If $\Delta_t \ge n_t/4$ then $X_t \le n_t/4$. Since $(x-k)/(y-k) \le x/y$ for any $k \ge 0$ and any $x,y >k$ with $y \ge x$, it holds for round $t+1$ that
\begin{eqnarray*}
  \Pr[ x_u=0 \mid u \mbox{ defined}] & = & \binom{3}{2} \frac{6X_t}{6n_t} \cdot \frac{6X_t-1}{6n_t-1} \cdot \frac{6Y_t}{6n_t-2} 
   +\ \frac{6X_t}{6n_t} \cdot \frac{6X_t-1}{6n_t-1} \cdot \frac{6X_t-2}{6n_t-2} \\
  & \le &
  \binom{3}{2} p_t^2 \cdot \frac{6Y_t}{6n_t-2} + p_t^3
\end{eqnarray*}
Moreover,
\[
  \frac{6Y_t}{6n_t-2} = (1+O(1/n_t))\frac{n_t-X_t}{n_t} = (1+O(1/n_t))(1-p_t)
\]
Hence,
\[
  \Pr[ x_u=0 \mid u \mbox{ defined}] \le 3 p_t^2 (1-p_t) \cdot(1+O(1/n_t))  + p_t^3 \le 3p_t^2
\]
as long as $X_t = \Omega(\log n)$ and $n$ is sufficiently large, which implies that
\[
  \E[ X_{t+1} \mid X_t, n_t, n_{t+1}] \le 3p_t^2 n_{t+1}
\]
For any defined node $v$ let the binary random variable $Z_v$ be 1 if and only if $x_v=0$, and let $N_t$ denote the set of all defined nodes. Then $X_{t+1}$ can be represented as
\[
  X_{t+1} = \sum_{v \in N_t} Z_v
\]
Note that the $Z_v$'s are not independent. However, the following claim can be shown.

\begin{claim} \label{cl:neg-cor2}
For any subset $S$ of defined nodes with $|S|\le n_t/4$ and any defined node $u \not\in S$ it holds for round $t+1$ that
\[
  \Pr[ x_u=0 \mid \bigwedge_{v \in S} x_v=0 ] \le 3.5 p_t^2
\]
\end{claim}
\begin{proof}
Consider any event where $x_v$ is set to 0 for all $v \in S$ and the nodes in $S$ inspected $k$ $0$s and $\ell$ $1$s for their decisions. Let $s=|S|$. Under the assumption that $u$ is defined,
\begin{eqnarray*}
  \Pr[ x_u \mbox{ is set to 0}] & = & \binom{3}{2} \frac{6X_t-k}{6n_t-3s} \cdot \frac{6X_t-k-1}{6n_t-3s-1} \cdot \frac{6Y_t-\ell}{6n_t-3s-2}  
   +\ \frac{6X_t-k}{6n_t-3s} \cdot \frac{6X_t-k-1}{6n_t-3s-1} \cdot \frac{6X_t-k-2}{6n_t-3s-2} \\
  & \le &
  \binom{3}{2} p_t^2 \cdot \frac{6Y_t-\ell}{6n_t-3s-2} + p_t^3
\end{eqnarray*}
In order to show the last inequality, we first note that $k+\ell=3s$ and $k \ge 2\ell$ since otherwise there must be a node $v\in S$ that saw at least two $1$s, which cannot happen. Hence, $k \ge 2s$. Because $X_t \le n_t/4$, this implies that $(6X_t-k-i)/(6n_t-3s-i) \ge (6X_t-i)/(6n_t-i)$ for any $i \in \{0,1,2\}$. On the other hand, since $s \le n_t/4$,
\[
  \frac{6Y_t-\ell}{6n_t-3s-2} \le 1.15 \frac{6Y_t}{6n_t-2} = (1.15 + O(1/n_t)) \frac{6Y_t}{6n_t}
\]
Hence,
\[
  \Pr[ x_u \mbox{ is set to 0}] \le \binom{3}{2} p_t^2 \cdot (1.15+O(1/n_t))(1-p_t) + p_t^3
  \le 3.5 p_t^2
\]
if $n_t$ is large enough, which proves the claim.
\end{proof}

Therefore, we can make use of the Chernoff bounds in Lemma~\ref{lem:chernoff} to show that $X_{t+1} \le 4 p_t^2 n_{t+1}$ w.h.p., as long as $\E[X_{t+1}] = \Omega(\log n)$. We know that $p_{t_0} \le 1/4$. Hence, for any $i \ge 0$ with $\E[X_{t_0+i}] = \Omega(\log n)$, $p_{t_0+i} \le (2p_{t_0})^{2^i}$, which means after $O(\log \log n)$ rounds we are down to just $X_t = O(\log n)$. Thus, we have reached an almost-everywhere consensus, and this will be preserved as our arguments imply that also in polynomially many subsequent rounds, $X_t = O(\log n)$ w.h.p.
\end{proof}

\subsubsection*{Case 2: {\boldmath $c \sqrt{n_t \ln n_t} \le \Delta_t <
n_t/4$} for a sufficiently large constant {\boldmath $c$}}
\newcommand{\lemcasetwo}{
If there is a round $t_0$ with $c \sqrt{n_{t_0} \ln n_{t_0}} \le \Delta_{t_0} \leq n_{t_0}/4$ for a sufficiently large constant $c$, then there is a round $t_1=t_0+O(\log n)$ with $\Delta_{t_1} \geq n_t/4$ w.h.p.
}
\begin{lemma} \label{lem:casetwo}
\lemcasetwo
\end{lemma}
\begin{proof}
Suppose that $\Delta_{t_0} \ge c \sqrt{n_{t_0} \log n_{t_0}}$. Let $\delta_t=\Delta_t/n_t$. Then it holds that $X_t = n_t/2-\Delta_t = n_t(1/2-\delta_t)$. Similar to the previous case, for any defined node $u$,
\begin{eqnarray*}
  \Pr[ x_u = 0] & = & \binom{3}{2} \frac{6X_t}{6n_t} \cdot \frac{6X_t-1}{6n_t-1} \cdot \frac{6Y_t}{6n_t-2} 
   +\ \frac{6X_t}{6n_t} \cdot \frac{6X_t-1}{6n_t-1} \cdot \frac{6X_t-2}{6n_t-2} \\
  & \le &
  \binom{3}{2} (1/2-\delta_t)^2 \cdot \frac{6Y_t}{6n_t-2} + (1/2-\delta_t)^3 \\
  & = & 3 (1/2-\delta_t)^2 (1/2+\delta_t) \cdot (1+O(1/n_t))
   +\ (1/2-\delta_t)^3 \\
  & = & 1/2 - (3/2)\delta_t +2\delta_t^3 + O(1/n_t)
\end{eqnarray*}
Since $c \sqrt{(\ln n_t)/n_t} \le \delta_t \le 1/4$, it follows that
\[
  \E[ X_{t+1} \mid X_t, n_t, n_{t+1}] \le (1/2-(11/8)\delta_t) n_{t+1}
\]
Like in the case of $\Delta_t \ge n_t/4$, $X_{t+1}$ can be represented as
\[
  X_{t+1} = \sum_{v \in N_t} Z_v
\]
for some binary random variables $Z_v$ that, unfortunately, are not independent. However, the following claim holds:

\begin{claim} \label{cl:neg-cor3}
For any subset $S$ of defined nodes with $|S|\le n_t/2$ and any defined node $u \not\in S$ it holds that
\[
  \Pr[ x_u=0 \mid \bigwedge_{v \in S} x_v=0 ] \le 1/2-(11/8)\delta_t
\]
\end{claim}
\begin{proof}
Consider any event $E(k,\ell)$ where $x_v$ is set to 0 for all $v \in S$ and the nodes in $S$ inspected $k$ Os and $\ell$ 1s for their decisions. Let $s=|S|$. Recall that $k+\ell=3s$ and $k \ge 2\ell$ and therefore, $k \ge 2s$. This means that the ratio between the number of remaining 0s and the remaining 1s is lower for $u$ than it was initially, so it should lower the probability that $x_u=0$. In fact, it already follows from the proof of Claim~\ref{cl:neg-cor2} that the function
\[
  f(s) = \frac{6X_t-k}{6n_t-3s} \cdot \frac{6X_t-k-1}{6n_t-3s-1} \cdot \frac{6X_t-k-2}{6n_t-3s-2}
\]
is monotonically decreasing for $s \ge 0$ if $k \ge 2s$. Via tedious calculations one can also show that the function
\[
  g(s) = \frac{6X_t-k}{6n_t-3s} \cdot \frac{6X_t-k-1}{6n_t-3s-1} \cdot \frac{6(n_t-X_t)-(3s-k)}{6n_t-3s-2}
\]
is monotonically decreasing for $s \ge 0$ if $k \ge 2s$, which implies that we can upper bound
$\Pr[ x_u=0 \mid \bigwedge_{v \in S} x_v=0 ]$ in the same way as for the isolated case $\Pr[x_u=0]$.
\end{proof}

Hence, it follows from the Chernoff bounds that $X_{t+1} \le (1/2-(9/8)\delta_t)n_{t+1}$ w.h.p. (by setting $\delta=3\delta_t/16$ in Lemma~\ref{lem:chernoff}) as long as $c \sqrt{n_t \ln n_t} \le \Delta_t < n_t/4$, i.e., $c \sqrt{(\ln n_t)/n_t} \le \delta_t \le 1/4$. This implies that $\delta_{t+1} \ge (9/8) \delta_t$ w.h.p. Thus, there is a round $t_1=t_0+O(\log n)$ with $\delta_{t_1} \ge 1/4$, w.h.p., which completes the proof.
\end{proof}

\subsubsection*{Case 3: {\boldmath $|\Delta_t| < c \sqrt{n_t \ln n_t}$}}

Up to this point we assumed that $X_t \le Y_t$ w.l.o.g. because once the difference between the number of $0$s and $1$s is at least $c \sqrt{n_t \ln n_t}$, the majority value will win w.h.p. However, if $|\Delta_t| < c \sqrt{n_t \ln n_t}$, it may not be clear any more who will win. First, we show that with constant probability it just takes a constant number of rounds until
$|\Delta_t| \ge c' \sqrt{n_t}$, and then we use that insight together with insights from Case 2 and a result about Markov chains to argue that within $O(\log n)$ rounds, $|\Delta_t| \ge c \sqrt{n_t \ln n_t}$ w.h.p. Before that, we need two technical lemmas.

We say that a random variable $Z$ {\em stochastically dominates} a random variable
$Z'$, and write $Z \succeq Z'$, if $\Pr[Z \ge x] \ge \Pr[Z' \ge x]$ for any $x$.
Recall that $\Delta_t = (Y_t-X_t)/2$, where $X_t$ is the number of $0$s and $Y_t$ is the number of $1$s. For simplicity we assume that $n_t$ is even so that $\Delta_t$ is an integer.

\newcommand{\lemdom}{
For any two imbalances $\Delta_t$ and $\Delta'_t$ with $\Delta_t \ge
\Delta'_t \ge 0$ it holds that $\Delta_{t+1} \succeq \Delta'_{t+1}$.
}
\begin{lemma}\label{lem:dom}
\lemdom
\end{lemma}
\begin{proof}
We show stochastic domination for any two imbalances $\Delta_t$ and
$\Delta'_t=\Delta_t-1$. The rest follows by induction. Let $z = n/2-\Delta'_t$.
Without loss of generality, we assume that nodes 1 to $z-1$ have value 0 in both $\Delta_t$ and $\Delta'_t$, and nodes $z+1,\ldots,n$ have value 1 in both $\Delta_t$ and $\Delta'_t$. Node $z$ has value 1 in $\Delta_t$ and value 0 in $\Delta'_t$.

Let $\Omega$ be the space of all possible outcomes of the random experiment of sending 6 copies of each defined value to nodes chosen independently and uniformly at random and making decisions at those nodes that receive more than 3 values which of these 3 to pick. Consider any such outcome $w \in \Omega$.

Any node that sets its value to 0 in $w$ assuming $\Delta$ also sets its value to 0 in $w$ assuming $\Delta'$. On the other hand, any node that sets its value to 1 in $w$ assuming $\Delta'$ also sets its value to 1 in $w$ assuming $\Delta$. So $X_{t+1}|_{\Delta}$ dominates $X_{t+1}|_{\Delta'}$ and $Y_{t+1}|_{\Delta'}$ dominates $Y_{t+1}|_{\Delta}$. Hence, $(Y_{t+1}-X_{t+1})|_\Delta$ dominates $(Y_{t+1}-X_{t+1})|_{\Delta'}$, which proves the lemma.
\end{proof}

The main lemma of this section is the following:
\newcommand{\lemconstantjump}{
For any $|\Delta_t| \ge 0$, $\Pr[|\Delta_{t+1}| \ge \sqrt{n_t/16}] \ge \alpha$ for some constant $\alpha>0$, provided $n_t$ is large enough.
}
\begin{lemma} \label{lem:constantjump}
\lemconstantjump
\end{lemma}

To prove this lemma, we use the Paley-Zygmund inequality (\Cref{lem:paley}) defined as follows:
\begin{lemma}[Paley-Zygmund] \label{lem:paley}
If $X \ge 0$ is a random variable with finite variance, then for any $0 \le \epsilon \le 1$,
\[
  \Pr[X > \epsilon \E[X] ] \ge (1-\epsilon)^2 \frac{\E[X]^2}{\E[X^2]}
\]
\end{lemma}

We now prove Lemma~\ref{lem:casethree}:
\begin{proof}
We only need to prove the lemma for $\Delta_t = 0$ because the general case follows from stochastic domination (see Lemma~\ref{lem:dom}).

Since for $\Delta_t = 0$ every defined node $v$ in round $t+1$ has the same probability of having $x_v=0$ or $x_v=1$, $\E[\Delta_{t+1}] = 0$. For any defined node $v$ in round $t+1$ let the binary random variable $Z_v$ be 1 if and only if $x_v$ is set to 0 in round $t+1$. In this case, given $N_{t+1}$,
\begin{eqnarray*}
  \Delta_{t+1} & = & \frac{1}{2} \left( \left( n_{t+1} - \sum_{v \in N_{t+1}} Z_v \right) - \sum_{v \in N_{t+1}} Z_v \right) 
    =   n_{t+1}/2 - \sum_{v \in N_{t+1}} Z_v
\end{eqnarray*}
From the definition of the variance it follows that
\[
  V[\Delta_{t+1} \mid N_{t+1}] = \sum_{v \in N_{t+1}} V[Z_v] + \sum_{v,w \in N_{t+1}, v\not=w} cov(Z_v,Z_w)
\]
For any $v \in N_{t+1}$, $V[Z_v] = \E[Z_v^2] - \E[Z_v]^2$. Furthermore, $\E[Z_v] = \Pr[Z_v=1] = 1/2$ and $\E[Z_v^2] = \Pr[Z_v=1] = 1/2$ as well. Hence, $V[Z_v] = 1/2-1/4 = 1/4$. Moreover,
\begin{align*}
  cov(Z_v,Z_w) & = \E[(Z_v-\E[Z_v])(Z_w-\E[Z_w])] = \E[Z_v \cdot Z_w] - \E[Z_v] \cdot \E[Z_w]
\end{align*}
Certainly,
\begin{align*}
  \E[Z_v \cdot Z_w] & = \Pr[Z_v=Z_w=1] = \Pr[Z_v=1] \cdot \Pr[Z_w=1 \mid Z_v=1]  = \frac{1}{2} \Pr[Z_w=1 \mid Z_v=1]
\end{align*}
but obtaining a good estimate of $\Pr[Z_w=1 \mid Z_v=1]$ is difficult, so we first note that
\[
  \Pr[v \mbox{ sees } (0,0,1)] = \frac{1}{8} + \frac{1}{48 n_t} \pm O \left( \frac{1}{n_t^2} \right)
\]
and
\[
  \Pr[v \mbox{ sees } (0,0,0)] = \frac{1}{8} - \frac{1}{16 n_t} \pm O \left( \frac{1}{n_t^2} \right)
\]
Next, we estimate $\Pr[Z_w=1]$ for the case that $v$ sees $(0,0,1)$ or $(0,0,0)$. First of all,
\begin{align*}
  \Pr[Z_w=1 \mid v \mbox{ sees } (0,0,1)] = & \binom{3}{2} \frac{6X_t-2}{6n_t-3} \cdot \frac{6X_t-3}{6n_t-4} \cdot \frac{6Y_t-1}{6n_t-5} 
   + \frac{6X_t-2}{6n_t-3} \cdot \frac{6X_t-3}{6n_t-4} \cdot \frac{6X_t-4}{6n_t-5}
\end{align*}
Furthermore, for $X_t=Y_t=n_t/2$,
$(6X_t-2)/(6n_t-3) = (6X_t)/(6n_t) - 1/(12n_t)) \pm O(1/n_t^2)$,
$(6X_t-3)/(6n_t-4) = (6X_t-1)/(6n_t-1) - 1/(12n_t)) \pm O(1/n_t^2)$,
$(6X_t-4)/(6n_t-5) = (6X_t-2)/(6n_t-2) - 1/(12n_t)) \pm O(1/n_t^2)$, and
$(6Y_t-1)/(6n_t-5) = (6Y_t)/(6n_t-2) + 1/(12n_t)) \pm O(1/n_t^2)$.
Thus,
\begin{align*}
  \Pr[&Z_w=1 \mid v \mbox{ sees } (0,0,1)] \\
  = & \Pr[Z_w = 1] - \binom{3}{2} \cdot \left( 2 \left( \frac{1}{2} \right)^2 \left( -\frac{1}{12n_t} \right) + \left( \frac{1}{2} \right)^2 \frac{1}{12n_t} \right)
  - 3 \cdot \left( \frac{1}{2} \right)^2 \cdot \frac{1}{12n_t} \pm O \left( \frac{1}{n_t^2} \right) \\
  = & \frac{1}{2} - \frac{1}{8 n_t} \pm O \left( \frac{1}{n_t^2} \right)
\end{align*}
On the other hand,
\begin{align*}
  \Pr[Z_w=1 \mid v \mbox{ sees } (0,0,0)] = & \binom{3}{2} \frac{6X_t-3}{6n_t-3} \cdot \frac{6X_t-4}{6n_t-4} \cdot \frac{6Y_t}{6n_t-5}
   + \frac{6X_t-3}{6n_t-3} \cdot \frac{6X_t-4}{6n_t-4} \cdot \frac{6X_t-5}{6n_t-5}
\end{align*}
and for $X_t=Y_t=n_t/2$,
$(6X_t-3)/(6n_t-3) = (6X_t)/(6n_t) - 1/(4n_t) \pm O(1/n_t^2)$,
$(6X_t-4)/(6n_t-4) = (6X_t-1)/(6n_t-1) - 1/(4n_t) \pm O(1/n_t^2)$,
$(6X_t-5)/(6n_t-5) = (6X_t-2)/(6n_t-2) - 1/(4n_t) \pm O(1/n_t^2)$, and
$(6Y_t)/(6n_t-5) = (6Y_t)/(6n_t-2) + 1/(4n_t) \pm O(1/n_t^2)$.
Thus, in this case,
\begin{align*}
  \Pr[Z_w=1 \mid v \mbox{ sees } (0,0,0)] & = \Pr[Z_w = 1] - \binom{3}{2} \cdot \frac{1}{4} \cdot \frac{1}{4n_t} - 3 \cdot \frac{1}{4} \cdot \frac{1}{4n_t} \pm O \left( \frac{1}{n_t^2} \right) \\
  & = \frac{1}{2} - \frac{3}{8 n_t} \pm O \left( \frac{1}{n_t^2} \right)
\end{align*}
Hence, altogether,
\begin{align*}
  \E[Z_v \cdot Z_w] &= 3 \cdot \left(  \frac{1}{8} + \frac{1}{48 n_t} \right)
  \left( \frac{1}{2} - \frac{1}{8 n_t} \right) 
  + \left( \frac{1}{8} - \frac{1}{16 n_t} \right) \left( \frac{1}{2} - \frac{3}{8 n_t} \right) \pm \ O \left( \frac{1}{n_t^2} \right) 
  &= \frac{1}{4} - \frac{3}{32 n_t} \pm O \left( \frac{1}{n_t^2} \right)
\end{align*}

and therefore,
\[
  cov(Z_v,Z_w) = - \frac{3}{32 n_t} \pm O \left( \frac{1}{n_t^2} \right)
\]
Plugging this into the formula for the variance, and using the fact that w.h.p. $n_t \in [3n/4,n]$ and therefore, $n_{t+1}/n_t \in [3/4,4/3]$, we get
\begin{align*}
  V[\Delta_{t+1} \mid N_{t+1}] & = \frac{n_{t+1}}{4} - \frac{3n_{t+1}(n_{t+1}-1)}{32 n_t} \pm O \left( \frac{n_{t+1}^2}{n_t^2} \right) \\
  & \ge \frac{n_{t+1}}{8} \pm O(1)
\end{align*}
Moreover, it follows from the definition of the variance that
\[
  V[\Delta_{t+1}] = \E[\Delta_{t+1}^2] - \E[\Delta_{t+1}]^2
\]
and since $\E[\Delta_{t+1}]=0$,
\[
  \E[|\Delta_{t+1}|^2] = \E[\Delta_{t+1}^2] \ge \frac{n_{t+1}}{8} \pm O(1)
\]

To calculate $\E[|\Delta_{t+1}|^4]$, note that
\begin{align*}
\E[|\Delta_{t+1}|^4] = & \E[(\frac{n_{t+1}}{2} - \sum_{v \in N_{t+1}}Z_v)^4] \\
 = & \left(\frac{n_{t+1}}{2}\right)^4 - 4\left(\frac{n_{t+1}}{2}\right)^3 \sum_{v \in N_{t+1}}\E[Z_v] \\
  & + 6\left(\frac{n_{t+1}}{2}\right)^2 \E[(\sum_{v \in N_{t+1}}Z_v)^2]  \\
  & - 4\cdot \frac{n_{t+1}}{2}\E[(\sum_{v \in N_{t+1}}Z_v)^3] + \E[(\sum_{v \in N_{t+1}}Z_v)^4]
\end{align*}
Furthermore,

\begin{eqnarray*}
 \sum_{v \in N_{t+1}}\E[Z_v] & = & \frac{n_{t+1}}{2} \\
 \E[(\sum_{v \in N_{t+1}}Z_v)^2] & = & \sum_{v \in N_{t+1}}\E[Z_v] + \sum_{v, w \in N_{t+1}, v \ne w} E[Z_v \cdot Z_w] \\
 \E[(\sum_{v \in N_{t+1}}Z_v)^3] & = & \sum_{v \in N_{t+1}}\E[Z_v] \\
   & & +\ \binom{3}{2} \sum_{v, w \in N_{t+1}, v \ne w} E[Z_v \cdot Z_w] \\
   & & +\ \sum_{u,v,w \in N_{t+1}, u \ne v \ne w} E[Z_u \cdot Z_v \cdot Z_w] \\
 \E[(\sum_{v \in N_{t+1}}Z_v)^4] & = & \sum_{v \in N_{t+1}}\E[Z_v] \\
   & & +\ \left(\binom{4}{3} + \frac{4!}{2!\cdot2!\cdot2}\right)\cdot \sum_{v, w \in N_{t+1}, v \ne w} \E[Z_v \cdot Z_w] \\
    & & +\ \binom{4}{2}\sum_{u,v,w \in N_{t+1}, u \ne v \ne w} E[Z_u \cdot Z_v \cdot Z_w] \\
    & & +\ \sum_{u,v,w,x \in N_{t+1}, u \ne v \ne w \ne x} E[Z_u \cdot Z_v \cdot Z_w \cdot Z_x]
\end{eqnarray*}
Putting this together and using a CAS, one can show that the maximum of $\E[|\Delta_{t+1}|^4]$ is achieved for $n_{t+1}/n_t = 4/3$ and in this case it is upper bounded by:
\[
 \E[|\Delta_{t+1}|^4] \leq \frac{1}{12} {n_{t+1}}^{2}+{\frac {5}{216}}n_{t+1}+{\frac {77}{5184}}+{\frac {371}{31104 \cdot n_{t+1}}} + O\left(\frac{1}{n_{t+1}^2}\right)
\]

As an intuition why $\E[|\Delta_{t+1}|^4] = O(n_{t+1}^2)$ holds, note that if the $Z_v$'s are independent, then one can make use of the formulas for higher moments of binomial distributions (see, e.g., \cite{JK77}) to show that $\E[|\Delta_{t+1}|^4] = O(n_{t+1}^2)$. In addition to that, as indicated by the negative covariance, seeing a majority of $0$s (resp. $1$s) in some nodes implies seeing fewer $0$s (resp. $1$s) for subsequent nodes, so one expects higher moments in our case to be less than in the independent case. In fact, in the independent case, $\E[|\Delta_{t+1}|^2] = n_{t+1}/4$.

Hence, it follows from the Payley-Zygmund inequality that
\[
  \Pr[|\Delta_{t+1}|^2 \ge \frac{n_{t+1}}{16}] \ge \alpha
\]
for some constant $\alpha>0$, which completes the proof of the lemma.
\end{proof}

Moreover, from the proof of Lemma~\ref{lem:casetwo} and the Chernoff bounds in Lemma~\ref{lem:chernoff} we get:
\begin{lemma} \label{lem:constantdrift}
For any $\delta_{t} = \Omega(1/\sqrt{n_t})$,
$
  \Pr[ \delta_{t+1} \geq (9/8) \delta_t ]
  \geq 1 - \exp \left(- \Theta( \delta_{t}^2 n_{t+1} )  \right).
$
\end{lemma}

Now we can finish Case 3.

\newcommand{\lemcasethree}{
If at a round $t_0$ we have $|\Delta_{t_0}| < c \sqrt{n_{t_0} \ln n_{t_0}}$ for the value of
$c$ needed by Lemma \ref{lem:casetwo}, then there is a round $t_1=t_0+O(\log n)$ with $|\Delta_{t_1}| \geq c \sqrt{n_t \ln n_t}$ w.h.p.
}
\begin{lemma} \label{lem:casethree}
\lemcasethree
\end{lemma}
\begin{proof}
Lemma \ref{lem:constantjump} implies that the expected number of steps until
we are in the hypothesis of Lemma \ref{lem:constantdrift} is $O(1)$. That
is, $|\Delta_{t}| = \gamma \sqrt{n_t}$ for some constant $\gamma>0$. Assume w.l.o.g. that $\Delta_t>0$, and consider $\Upsilon_t = \lfloor \delta_t \sqrt{n_t}/\gamma \lfloor$ and let $q=\lfloor c\sqrt{\ln n_t}/\gamma \rfloor$ denote the maximum value of $\Upsilon_t$. To continue, we need the following technical result from \cite{doerr2011adversary}.

\begin{claim}\label{cl:dichotomy}
Let $(X_t)_{t=1}^{\infty}$ be a Markov Chain with state space $\{0, \ldots,
q\}$ that has the following properties:
\begin{itemize}
 \item there are constants $c_1 > 1$ and $c_2>0$, such that for any $t \ge 1$,
$ \Pr[X_{t+1} \geq \min\{c_1 X_{t},q\} ] \geq 1 - e^{-c_2 X_t}$,
 \item $X_{t}=0 \Rightarrow X_{t+1} \ge 1$ with probability $c_3$ which is a constant greater than $0$.
\end{itemize}
Let $c_4 > 0$ be an arbitrary constant and $T:= \min \{ t \ge 1 \colon X_t
\geq c_4 \log q \} $. Then for every constant $c_6 > 0$ there is a constant
$c_5=c_5(c_4,c_6) > 0$ such that
\[
  \Pr[ T \leq c_5 \cdot \log q + \log_{c_1} (c_4 \log q) ] \geq 1-q^{-c_6}.
\]
\end{claim}

By this claim, $O(\log q)$ rounds suffice to achieve $\Upsilon_t \geq c_4
\log q$, or $\Delta_{t+\tau - 1}\geq c\sqrt{n_{t+\tau - 1}} \cdot c_4 \log q$, w.h.p. for
any constant $c_4>0$, which finishes the proof.
\end{proof}

\noindent\textbf{Deciding on a consensus value.}
Note that while the $(k,\ell)$-majority algorithm converges to a common plurality value within $O(\log n)$ rounds, the nodes cannot determine whether an almost-everywhere consensus has already been reached.
We address this by executing the \textbf{deciding $(k,\ell)$-majority algorithm} in which each node $u$ executes both the reset rule and the update rule of the original $(k,\ell)$-majority algorithm and, additionally, the following rule:
\begin{compactitem}
\item \textbf{decision rule:} If there is some $y \in \{0,1\}$ so that for the past $\alpha \ln n$ rounds, $x_u \in \{y,\bot\}$ and for at least half of them, $x_u=y$, then $u$ outputs $y$.
\end{compactitem}
Here, $\alpha > 0$ is a constant determined in the proofs of the following two lemmas. 

\newcommand{\thmextensionone}{
For any initial assignment of values to the nodes with no defined node and any $\epsilon \le 1/16$, there is a constant $\beta>0$ so that every node that is blocked in at most a fifth of $\beta \ln n$ rounds will output a value within these rounds, w.h.p.
}
\begin{lemma}\label{thm:extension1}
\thmextensionone
\end{lemma}
\begin{proof}
We know that there is a constant $\beta'>0$ so that the original $(k,\ell)$-majority algorithm reaches an almost-everywhere consensus on some value $y \in \{0,1\}$ in at most $\beta' \ln n$ rounds, w.h.p. Consider a node $u$ that is blocked in at most a fifth of $\beta \ln n$ rounds, where $\beta \ge 5 \beta'$. Then $u$ is blocked in at most a fourth of the last $(\beta-\beta') \ln n$ rounds, since otherwise it would be blocked too often over all of the $\beta \ln n$ rounds. Next, assume that $\beta$ is large enough so that the last $(\beta-\beta')\ln n$ rounds can be cut into $c$ intervals of $\alpha \ln n$ rounds each. For at most $2c/3$ of these intervals $u$ is blocked at least $1/3$ of the rounds since otherwise it would again be blocked too often. Let $S$ be the set of the remaining $c/3$ intervals. Consider some fixed interval $I \in S$. Since $u$ is blocked at most a third of the rounds in $I$ and for any round in which $u$ is non-blocked, $\Pr[u$ is undefined$] \le 2/11$, it follows from standard Chernoff bounds that there are at least half of the rounds in $I$ where $u$ is defined w.h.p. (if $\alpha>0$ and $n$ are large enough). Thus, any interval $I \in S$ is a candidate where $u$ may output a value. For any round $t$ in such an $I$ in which $u$ is defined, the probability that $u$ sets its variable to the non-consensus value is at most
\[
  \binom{3}{2} \left( \frac{\gamma \log n}{3n/4} \right)^2 + \left( \frac{\gamma \log n}{3n/4} \right)^3 \le \left( \frac{3\gamma \log n}{n} \right)^2
\]
Hence, the probability that $u$ sets its variable to the non-consensus value in some round in $I$ is at most
\[ 
  \alpha \ln n \cdot \left( \frac{3\gamma \log n}{n} \right)^2 < \frac{1}{2n}
\]
if $n$ is sufficiently large, and therefore, $\Pr[u$ does not output $y$ at the end of $I] \le 1/n$.
Thus, the probability that there is no $I \in S$ at which $u$ outputs $y$ at the end of $I$ is at most $n^{-c/3}$, which completes the proof.
\end{proof}

Note that a $1/16$-bounded adversary can cause at most $1/3$ of the nodes to be blocked more than $1/5$ of the rounds, which implies that at least $2/3$ of the nodes will produce an output w.h.p. Finally, we show that it cannot happen that two nodes output different values, so nodes eventually decide such that almost-everywhere consensus holds w.h.p.

\newcommand{\thmextensiontwo}{
There is a constant $\alpha > 0$ such that if a node $u$ outputs a value $y \in \{0, 1\}$ in the deciding $(k,\ell)$-majority algorithm, then no node ever outputs a value $z \ne y$ w.h.p.
}
\begin{lemma}\label{thm:extension2}
\thmextensiontwo
\end{lemma}
\begin{proof}
First, we show that if a node $u$ outputs a value $y \in \{0,1\}$, then there must have been a round $t$ in the past $\alpha \ln n$ rounds with $|\Delta_t| \ge n_t/4$. Suppose that there was no such round. Then for every round $t$ of these $\alpha \ln n$ rounds there are at least $n_t/4$ nodes $v \in V$ with $x_v=0$ and at least $n_t/4$ nodes $v \in V$ with $x_v=1$. Let $T$ be the set of the past $\alpha \ln n$ rounds in which $u$ was defined. We know that if $u$ outputs a value then $|T| \ge (\alpha/2) \ln n$. For any round $t \in T$, 
\[
  \Pr[ x_u=0 \mid \mbox{$u$ is defined} ] \ge \binom{3}{2} \left( \frac{1}{4} \right)^2 \left( 1- \frac{1}{4} \right) \pm \left( \frac{1}{n_t} \right) \ge \frac{1}{8}
\]
and the same lower bound also applies to $\Pr[x_u = 1 \mid u$ is defined $]$. Thus, for any $z \in \{0,1\}$,
\[
  \Pr[ x_u \not=z \mbox{ for all } t \in T] \le (1-1/8)^{(\alpha/2)\ln n}
  \le e^{-(\alpha/16)\ln n} = n^{-\alpha/16}
\]
This, however, means that if $\alpha>0$ is a sufficiently large constant, then w.h.p. $u$ would not output a value. Hence, if $u$ outputs a value, there must have been a round $t_0$ in the past $\alpha \ln n$ rounds with $|\Delta_{t_0}| \ge n_{t_0}/4$, w.h.p. 

Next we show that the value $y$ output by $u$ must be the majority value at round $t_0$. This holds because from $t_0$ there will always be at least $n_t/4$ of the nodes with value $y$, w.h.p., so the probability that $x_u$ is never equal to $y$ in the past $\alpha \ln n$ rounds is still polynomially small in $n$.

Suppose now that $u$ is the first node that outputs a value $y$. Then there cannot be a node $v$ that subsequently outputs a different value for the same reasoning as for $u$ since also for $v$ there must have been at least $n_t/4$ of the nodes with value $y$ in each of the past $\alpha \ln n$ rounds, completing the proof.
\end{proof}

Combining \Cref{thm:extension1} and \Cref{thm:extension2} with \Cref{thm:convergence}, and observing that Cases~1 and 2 imply the plurality property we have shown the following result:

\begin{theorem} \label{lem:binaryconsensus}
  The $(k,\ell)$-majority algorithm together with the decision rule  achieves binary almost-everywhere consensus (w.h.p.) in $O(\log n)$ rounds, against a late adversary who can block up to $\epsilon n$ nodes in each round, for $\epsilon < \tfrac{1}{16}$. 
  Moreover, the decision value will satisfy the plurality property (w.h.p.).
\end{theorem}

\section{Multi-Value Consensus} \label{sec:multivalue}
 
We now describe an algorithm for the case where nodes start out with arbitrary input values. 
While the algorithm achieves the same runtime bound as the majority rule algorithm presented in \Cref{sec:binary}, we point out that it does not guarantee the plurality property, due to its use of a maximum-rule for spreading a common value.

We first give a high-level overview and present the detailed pseudocode in Algorithm~\ref{alg:multivalue}. 
Initially, each node $u$ has a local variable $x_u$ that is set to an input value chosen by the adversary from a fixed domain of size $\poly(n)$. 
We say that a node is \emph{active in the current round}, if it initiates communication with other nodes and call it \emph{inactive} otherwise.
In the first round, each node becomes {active} with probability $\Theta(\frac{\log n}{n})$ and every active node sends its input value to a set of $\Theta(\log n)$ nodes chosen uniformly at random.
Only initially active nodes retain their input value; all other nodes (including the ones that are blocked) reset their value to a default value $\bot$, which is considered to be the minimum value. 

From this point onward, the algorithm performs $\Theta(\log n)$ iterations of information spreading: 
We start an iteration by instructing nodes that have received messages to mark themselves as active.
Then, every active node computes its new value as the maximum of the received values including its current value, 
sends its updated value to $2$ nodes $v$ and $w$ chosen uniformly at random.
In the next iteration, $v$ and $w$ in turn become active and participate in spreading their newly computed values.
After $\Theta(\log n)$ iterations, all nodes decide on their current value. 

To understand why the first round requires active nodes to communicate with $\Theta(\log n)$ nodes each, it is instructive to consider the ``uniform'' algorithm where active nodes always contact $2$ peers.  
Since the adversary can block up to $\epsilon n$ nodes in each round, it has constant probability of blocking \emph{all} nodes that currently hold the maximum value $x^*$, assuming it targets a randomly chosen subset of nodes during the first few rounds.  
If successful, the adversary can slowdown the spreading of value $x^*$ among the nodes in the network.
The end result of this scenario is that a large constant fraction of the nodes never receive $x^*$ and thus end up deciding on some other value $x'<x^*$ instead---a violation of almost-everywhere agreement.

\begin{algorithm}
\begin{algorithmic}[1]
\STATE Every node $u$ initalizes $x_u$ with its input value and, with probability $\tfrac{c_1\log n}{n}$, joins the \emph{active} nodes; otherwise it marks itself \emph{inactive}. \COMMENT{ Note that $c_1$, $c_2$, $c_3>0$ are constants.}
\STATE Each active node $u$ sends $x_u$ to a set of $\lceil c_2\log n\rceil$ uniformly at random chosen nodes.
\STATE Every inactive or blocked node $v$ sets $x_v = \bot$; (we define $\bot\le x$ for all values $x$).
\FOR{$t=1,\dots,\lceil c_3\log n\rceil$ iterations}
  \FOR{every node $v$}
    \STATE $v$ updates $x_v = \max(R \cup \{x_v\})$, where $R$ is the set of newly received values.
    \STATE \textbf{if} $x_v \ne \bot$ \textbf{then} $v$ joins \emph{active} nodes.
    \IF{$t<c_3\log n$ and $v$ is active} 
      \STATE Node $v$ sends $x_v$ to $2$ unif.\ at random chosen nodes.
    \ENDIF
  \ENDFOR
\ENDFOR
\STATE Every node $u$ decides on $x_u$ and terminates.
\end{algorithmic}
\caption{Multi-value consensus against a late adversary}
\label{alg:multivalue}
\end{algorithm}

\begin{theorem} \label{thm:multivalue}
Consider the late adversary that can block up to $\epsilon n$ nodes in each round, for $\epsilon\le\tfrac{1}{10}$, and suppose that nodes start with input values chosen by the adversary. 
Then, for any positive constant $\delta<1$, there is an algorithm that achieves consensus among $(1-\tfrac{\epsilon}{\delta})n$ nodes in $O(\log n)$ rounds with high probability, while sending $O(n\log n)$ messages in total. 
\end{theorem}

We first recall some concentration bounds. 

\begin{definition}[Martingale~\cite{mitzenmacher}] \label{def:martingale}
A sequence of random variables $Z_0,Z_1,\dots$ is a martingale with respect to the sequence $X_0,X_1,\dots$ if, for all $t \ge 0$, the following conditions hold:
\begin{compactenum}
\item $Z_t$ is a function of $X_0,\dots,X_t$.
\item $\E[Z_t] < \infty$.
\item $\E[ Z_{t+1} \mid X_0,\dots,X_t] = Z_t$.
\end{compactenum}
\end{definition}

\begin{definition}[Lipschitz Condition \cite{mitzenmacher}] \label{def:lipschitz}
Let $f(X_1,\dots,X_n)$ be a function to the reals. We say that $f$ satisfies the \emph{Lipschitz condition} with bound $c$ if, for any set of values $x_1,\dots,x_n$ and $y_i$, 
\[
  |f(x_1,\dots,x_{i-1},x_i,x_{i+1},\dots,x_n) - 
  f(x_1,\dots,x_{i-1},y_i,x_{i+i},\dots,x_n)| \le c.
\]
\end{definition}

\begin{theorem}[Azuma-Hoeffding Inequality \cite{mitzenmacher}] \label{thm:azuma}
  Let $X_0,\dots,$ be a martingale sequence such that for each $k$, $|X_k - X_{k-1}|\le c_k$. Then, for all $t\ge 0$ and all $\lambda>0$, 
  \[
    \prob\left[ |X_t - X_0| \ge \lambda \right]  \le 2\exp\left( - \frac{\lambda^2}{2\sum_{k=1}^t c_k^2} \right).
  \]
\end{theorem}

\noindent\textbf{Proof of \Cref{thm:multivalue}.}
Since the late adversary can block a constant fraction of nodes, we cannot rely on existing results on information spreading to show almost-everywhere agreement.
For round $t$, let $B_t$ be the set of blocked nodes and define $A$ to be the set of active nodes in round $1$. 
Let $x^*$ be the maximum value of all input values of nodes in $A \setminus B_1$ and let $S_0 \subseteq A \setminus B_1$ be the set of non-blocked active nodes\footnote{Recall that the nodes that the adversary blocks are affected \emph{after} they have performed their local computation.} starting with $x^*$.

For rounds $t\ge 1$, we define $S_t$ to be the set of nodes whose value is $x^*$ at the end of round $t$ and define $\sigma_t = |S_t|$.
The following two lemmas show that $\sigma_t$ grows by a constant factor until it reaches a size of $\Omega(n)$.

\newcommand{\lemvone}{
  $\sigma_1 \ge \Omega(\log n)$.
}
\begin{lemma} \label{lem:v1}
  \lemvone
\end{lemma}
\begin{proof}
By description of the algorithm, the expected number of active nodes is $\ge c_1\log n$ and, since the adversary is oblivious to the current coin flips of the nodes, the probability of blocking an active node is at most $\epsilon$, resulting in at most $\epsilon \cdot c_1 \log n$ blocked active nodes in round $1$ in expectation. 
As nodes become active independently, a standard Chernoff bound ensures that, w.h.p., 
\[
  |A \setminus B_1| \ge \tfrac{1}{2}\left(1 - \epsilon\right)c_1 \log n = \Omega(\log n).
\] 
Since $x^*$ is defined as the maximum of the input values in $A \setminus B_1$, at least $c_1\log n$ randomly chosen peers are contacted by at least one active non-blocked node in $A\setminus B_1$ that has $x^*$ as its input value. 
Similarly to above, we use a Chernoff bound to argue that at least $\tfrac{1}{2}(1 - \epsilon)c_1\log n$ of the contacted peers are non-blocked in the current round and hence adopt $x^*$, thus proving $\sigma_1 = \Omega(\log n)$.
\end{proof}

\begin{lemma} \label{lem:vt}
Let $r = r(t) = \tfrac{\sigma_{t-1}}{n}$ and let $\rho<1$ be a positive constant determined in the proof.
There is a constant $\alpha>0$ such that, for any $t$ where $r < \rho$, we have
$  \sigma_t \ge \max\left\{(1+\alpha)\sigma_{t-1},\Theta(\log n)\right\}.$
\end{lemma}
\begin{proof}
We proceed by induction over $t$. 
To simplify our analysis, we assume that after having sent it to $2$ randomly chosen peers, active nodes ``forget'' their current value and adopt value $\bot$. Clearly, this assumption can only worsen the runtime of the algorithm.
In other words, $\sigma_{t}$ only accounts for nodes that have received a message with $x^*$ in round $t$.

The base case $t=1$ follows from \Cref{lem:v1}. 
Now consider any round $t>1$. By the inductive hypothesis, $S_{t-1}$ (i.e.\ the set of nodes currently holding $x^*$) contains at least $\max\left\{(1+\alpha)\sigma_{t-2},\Theta(\log n)\right\}$ nodes unknown to the adversary. 
Let $U \subseteq S_{t-1}$ be the subset of non-blocked nodes in round $t$ that hold $x^*$ and let 
\begin{align} \label{eq:beta}
\beta = \tfrac{7}{8}(1 - \epsilon).
\end{align}
Since $\sigma_{t-1} = |S_{t-1}| \ge c_1\log n$ for a sufficiently large constant $c_1>0$ (see \Cref{lem:v1}), we can apply a standard Chernoff bound to show that 
$
  |U| \ge \beta\cdot\sigma_{t-1}
$
with high probability. 
Let $m = 2|U|$ denote the number of messages sent by nodes in $U$.
Since $\epsilon<\tfrac{1}{2}$, the event $M$, which we define as
\begin{align} \label{eq:m_bound}
  m \ge (1 + \beta)\sigma_{t-1},
\end{align}
happens with high probability.

For the purpose of our analysis, we consider each of the $m$ messages as a \emph{ball} and the possible destinations of these messages as \emph{bins}.
Let $X_i$ be the indicator random variable that the $i$-th bin is hit by at least $1$ ball.
From \eqref{eq:m_bound} and the fact that $r = \tfrac{\sigma_{t-1}}{n}$, we get
\begin{align} 
  \label{eq:xi_prob_bound1}
  \prob[ X_i \!=\! 1 \mid M ] 
    &= 1 - \left(1 - \tfrac{1}{n}\right)^m 
    \ge 1 - e^{-(1 + \beta)r}. 
\end{align}
Consider the functions
  $f(r) = 1 - e^{-(1 + \beta)r}$ and
  $g(r) = \left(1 + \tfrac{\beta}{10}\right)r$,
and their respective derivatives $f'(r) = (1 + \beta)e^{-(1+\beta)r}$ and $g'(r) = 1 + \tfrac{\beta}{10}$. 
Let 
\begin{align} \label{eq:rho}
  \rho = \log\left( \frac{1 + \beta}{1+{\beta}/{10}}\right)/(1 + \beta)>\tfrac{1}{5}\ge 2\epsilon,
\end{align}
where we have used the assumption that $\epsilon\le\tfrac{1}{10}$ and \eqref{eq:beta}.
Since $f(0) = g(0)$ and $f'(r) > g'(r)$ for all $r \in [0,\rho]$, we know that $f$ is bounded from below by $g$. 
It follows from \eqref{eq:xi_prob_bound1} that
  $\prob\left[ X_i \!=\! 1 \mid M \right] \ge \left(1 + \tfrac{\beta}{10}\right)r.$
By definition, $\sigma_t = \sum_{i=1}^n X_i$, hence 
\begin{align}
  \label{eq:expect_bound}
  \E[\sigma_t \mid M] 
    &\ge \left(1 + \tfrac{\beta}{10}\right)\sigma_{t-1}. 
\end{align}
As the random variables $X_i$ are not independent, we cannot apply a simple Chernoff Bound argument to obtain a high concentration bound. 
Instead, we will set up a martingale and use the method of bounded differences (cf. Chap~12 in \cite{mitzenmacher}).
Let $Y_i$ be the bin chosen by the $i$-th ball and let $Z=h(Y_1,\dots,Y_n)$ be the number of nonempty bins.
For $1 \le i \le n$, we define
$
  Z_i = \E[ h(Y_1,\dots,Y_n) \mid Y_1, \dots, Y_{i}, M],
$
and let $Z_0 = \E[ h(Y_1,\dots,Y_n) \mid M]$.
To see that the sequence $Z_0,Z_1,\dots,Z_n$ is a martingale (cf.\ Def.~\ref{def:martingale}), note that, by definition of $Z_i$, we have
\begin{align*}
  \E&\left[Z_i \mid Y_1, \dots, Y_{i-1}, M \right] 
   = \E\left[\ \E\left[ Z \mid Y_1, \dots, Y_{i}, M \right] 
                \mid Y_1, \dots, Y_{i-1}, M\right]  
   =  \E\left[ Z \mid Y_1, \dots, Y_{i-1}, M \right] 
  = Z_{i-1}.
\end{align*}
As changing the destination of any ball can affect the number of nonempty bins by at most $1$, we know that the function $h$ satisfies the Lipschitz condition (see Def.~\ref{def:lipschitz}) with bound $1$.
Thus we can apply the Azuma-Hoeffding inequality (see \Cref{thm:azuma}) with parameter 
$\lambda = \sqrt{c'\ \sigma_{t-1}\log n}$, for a sufficiently large constant $c'>0$, to derive that
\begin{align*}
  \prob\left[\sigma_t \le \left(1 + \tfrac{\beta}{20}\right)\sigma_{t-1} \mid M\right]
  &\le 
  \prob\left[\sigma_t \le \left(1 + \tfrac{\beta}{10}\right)\sigma_{t-1} - \lambda \mid M\right] \\
  &\le 
  \prob\left[\sigma_t \le \E[\sigma_{t} \mid M] - \lambda \mid M\right] \tag{by \eqref{eq:expect_bound}}\\
  &\le 2\exp\left( - \frac{c'\ \sigma_{t-1}\log n}{2m} \right), \tag{by \Cref{thm:azuma}}
\end{align*}
which is $n^{-\Omega(1)}$ because of \eqref{eq:m_bound}.
Recalling that event $M$ happens w.h.p., we can remove the conditioning and conclude that
\begin{align*}
  \prob\left[\sigma_t \!\ge\! \left(1 + \tfrac{\beta}{20}\right)\sigma_{t-1} \right]
  & \ge \prob\left[\sigma_t \!\ge\! \left(1 + \tfrac{\beta}{20}\right)\sigma_{t-1} \mid M\right]\cdot \prob\left[M\right] 
   \ge  1 - n^{-\Omega(1)}. 
\end{align*}
\end{proof}

Consider iterations $t=1,\dots,\Theta(\log n)$ of the algorithm. 
As long as $r\le \rho$, \Cref{lem:vt} tells us that the current set of nodes that are adopting $x^*$ keeps growing by a constant factor with high probability in each iteration. We can guarantee that this happens (w.h.p.) by taking a union bound over the $\Theta(\log n)$ iterations of the algorithm until $r>\rho$.
In the following lemma, we handle the ``last mile'' of spreading value $x^*$ to (almost) all of the remaining $<(1 - \rho)n$ nodes.

\newcommand{\lemFinalPhase}{%
  Let $t_1$ be a round where $r = \frac{\sigma_{t_1}}{n}>\rho$. 
  Then, after additional $\Theta(\log n)$ iterations of the algorithm, all except at most $(1 - \tfrac{\epsilon}{\delta})n$ nodes adopt value $x^*$ with high probability, where $\delta$ is any constant $<1$. This holds even if the adversary is strongly adaptive.
}
\begin{lemma} \label{lem:final-phase}
  \lemFinalPhase
\end{lemma}
\begin{proof}
Once the number of nodes that adopted $x^*$ is already $\rho n>\epsilon n$, we do not need to make use of the lateness of the adversary in the proof; in fact, we assume that the adversary is strongly adaptive.
Conditioned on event $M$ and assuming that the adversary blocks up to $\epsilon n$ nodes, we know from \eqref{eq:m_bound} that the number of messages carrying $x^*$ is at least
\[ 
  (1 + \beta)\rho n - 2\epsilon n \ge \beta \epsilon n = \Omega(n),
\]
where we have used the fact that $\rho\ge 2\epsilon$ (cf.\ \eqref{eq:rho}).
Let $X_{i,t}=1$ if and only if 
bin $i$ is nonempty at the end of round $t$ and let $B_{i,t}$ be the indicator random variable that node $i$ is blocked by the late adversary in round $t$. 
Clearly, for any $t\ge t_1$, 
\begin{align} \label{eq:adopt}
  \prob[ X_{i,t}\!=\! 1 \mid B_{i,t}\!=\!0, M, X_{i,t-1}\!=\!0 ] 
    \ge 1 - \left( 1 - \frac{1}{n} \right)^{2\beta \epsilon n}
    \ge 1 - e^{-2\beta \epsilon}, 
\end{align}
which is a constant.
Since the adversary can block at most $\epsilon\cdot c_4\cdot n \log n$ nodes in total during $c_4\log n$ rounds, a simple counting argument tells us that, for any $\delta<1$, at least $(1 - \tfrac{\epsilon}{\delta})n$ nodes are unblocked during at least $(1 - \delta)c_4\log n$ rounds.
By choosing $c_4$ to be a sufficiently large constant depending on $\epsilon$, $\alpha$, and $\delta$, it follows from \eqref{eq:adopt} that the probability of bin $i$ remaining empty in all of the next $c_4\log n$ iterations is at most $n^{-\Omega(1)}$.
Finally, we can take a union bound over the $O(n)$ empty bins to show that at least $(1 - \tfrac{\epsilon}{\delta})n$ nodes adopt the value $x^*$ with high probability after the next $\Theta(\log n)$ iterations.
\end{proof}

From \Cref{lem:final-phase} it follows that value $x^*$ reaches all but a small constant fraction of the nodes. Since $x^*$ is the maximum value in the network after round 1 (cf.\ Line 3), the algorithm satisfies almost-everywhere agreement and the validity property as claimed, thus completing the proof of \Cref{thm:multivalue}.

\section{A Time Lower Bound for Agreement with Limited Communication} \label{sec:lowerbound}

We now study the impact of limited communication on the ability of nodes to solve almost everywhere consensus. In particular, we use the time-tested method of indistinguishable executions to provide a trade-off between the amount of communication available and the time required for termination correctly. 
We point out that these results hold even if the adversary does not block any nodes. 
Consequently, this means that increasing the lateness of the  adversary cannot improve the time complexity bounds of solving almost-everywhere consensus. 

\begin{theorem} \label{thm:lowerbound}
Suppose that each node can send a message to up to $d$ nodes in each round.
There is no randomized Monte Carlo algorithm that achieves almost-everywhere consensus with probability at least $1 - O(\tfrac{1}{n})$ among $>(1 - \epsilon')n$ nodes and terminates in $o(\log_{d} n)$ rounds, where $0<\epsilon'<\tfrac{1}{2}$. 
This holds even if no nodes are blocked by the adversary and each node has a unique identifier known to all other nodes.
\end{theorem}

Before proving \Cref{thm:lowerbound}, we consider its impact on the complexity bounds of our algorithms.
Since the majority-rule algorithm presented in \Cref{sec:binary} has $d = O(1)$, its time complexity bound is optimal. 
For the maximum propagation algorithm in \Cref{sec:multivalue} we have $d = O(\log n)$, thus applying \Cref{thm:lowerbound} tells us that the bound is optimal up to a factor of $O(\log\log n)$.


The overall structure of our proof follows the classical results of \cite{dolevstrong,CMS89} and a more recent generalization of these techniques to the asynchronous setting presented in \cite{AH2010}, with the important difference that here we consider almost everywhere agreement instead of consensus.
We first introduce some technical machinery.
We recall that an \emph{execution} of an algorithm is a sequence of rounds that is fully determined by the input value assignment of the nodes and $n$ (infinite) bit strings representing the nodes' private sources of randomness.
We use the notation $\alpha\overset{P}{\sim}\alpha'$ to say that executions $\alpha$ and $\alpha'$ are \emph{indistinguishable for a set of nodes $P$}, which means that each $u \in P$ goes through the same sequence of local state transitions in both executions. 
Intuitively speaking, this means that node $u$ cannot tell the difference between $\alpha$ and $\alpha'$.
A necessary and sufficient condition for this is that $u$ observes the same sequence of random bits and messages in both executions, and starts with the same input value.
We define an \emph{indistinguishability chain} as a sequence of executions $(\alpha_1,\alpha_2,\dots)$, where there are nonempty sets of nodes $P_1,P_2,\dots$ such that, for all $i$, it holds that $\alpha_i \overset{P_i}{\sim} \alpha_{i+1}$.
We refer the reader to Chapter 5 in \cite{attiyawelch} for a more complete definition.

We are now ready to prove \Cref{thm:lowerbound}.
For the sake of a contradiction, suppose that there is an algorithm $\cA$ that terminates in $t = o(\log_d n)$ rounds and achieves almost-everywhere consensus among all except $<\epsilon' n$ nodes with high probability.
For parameter
\begin{align} \label{eq:k}
  k = \frac{(1 - \epsilon')n}{d^t},
\end{align}
we consider a partitioning of the set $V$ of $n$ nodes into $\{S_1,\dots,S_{\lceil n/k\rceil}\}$ such that $|S_i|\le k$. 
For $0 \le i \le n/k$, we consider an execution $\alpha_i$ and assign the nodes in $\bigcup_{j=1}^i S_j$ input value $1$, where we follow the convention that $\bigcup_{j=1}^0S_j =\emptyset$.
All other nodes start with input value $0$ in $\alpha$.
This ensures that all nodes start with $0$ in $\alpha_0$ whereas all nodes start with $1$ in $\alpha_{\lceil n/k \rceil}$.
We will now show that $(\alpha_0,\dots,\alpha_{\lceil n/k \rceil})$ forms an indistinguishability chain with respect to large sets of nodes. 

\begin{lemma} \label{lem:indist1}
Suppose that each node is equipped with the same random bits in every $\alpha_i$.
For all $0\le i \le \lceil n/k\rceil$, there exists a set $P_i$ of size $\ge \epsilon' n$, such that $\alpha_i  \overset{P_i}{\sim} \alpha_{i+1}$ during the first $t = o(\log_d n)$ rounds of the algorithm.
\end{lemma}

\begin{proof}
Consider any $i$.
We observe that the only difference between the initial configurations of $\alpha_i$ and $\alpha_{i+1}$ are the input values for nodes in $S_{i+1}$, and this distinction is known to $|S_{i+1}| \le k$ distinct nodes initially.
Since each node can communicate with at most $d$ other nodes in each round, \eqref{eq:k} tells us that the total number of nodes that can be causally influenced (cf.\ \cite{lamport1978ordering}) by any node in $S_{i+1}$ is at most 
 $k d^t \le (1 - \epsilon')n$ during the first $t$ rounds of the algorithm.
 Thus, there exists a set $P_{i}$ of at least $\epsilon' n$ nodes for which $\alpha_i$ and $\alpha_{i+1}$ are indistinguishable.
\end{proof}
In the remainder of the proof, 
we use \Cref{lem:indist1} and adapt the arguments of \cite{AH2010}.
Let $F$ be the event that algorithm $\cA$ fails by violating almost-everywhere  agreement in at least one $\alpha_i$ in the chain. Note that this is the case when two sets of nodes of size $\ge \epsilon' n$ decide on conflicting values in $\alpha_i$.

A crucial consequence of \Cref{lem:indist1} is that, for each pair $\alpha_i$ and $\alpha_{i+1}$, there is a large set of $\ge \epsilon' n$ nodes that decide on the same value in both executions. 
By the validity property, we know that the decision is $0$ in $\alpha_0$ and $1$ in $\alpha_{\lceil n/k \rceil}$.
Therefore, considering that $\cA$ is a Monte Carlo algorithm that is guaranteed to terminate within $t$ rounds, we have $\prob[ F ] = 1$.

Recall that we have assumed in contradiction that algorithm $\cA$ succeeds with probability at least $1 - O(\tfrac{1}{n})$ and hence, almost-everywhere  agreement can be violated with probability $O(\tfrac{1}{n})$ in any individual execution.
By taking a union bound over the $\lceil \tfrac{n}{k} \rceil + 1$ executions in the chain, we get
\begin{align*}
  1 = \prob[ F ] 
  \le \left( \left\lceil \frac{n}{k} \right\rceil + 1\right) O\left(\frac{1}{n}\right) 
  = O\left(\frac{d^t}{(1 - \epsilon')n}\right) + O\left(\frac{1}{n}\right)  
  =o(1),  
\end{align*}
where we have used the bound \eqref{eq:k} and 
where the final equality follows from the premise $t = o(\log_d n)$.
This is a contradiction and completes the proof of \Cref{thm:lowerbound}.

\section{A Time Lower Bound for Agreement under the Strongly Adaptive Blocking Adversary} 
  \label{sec:aelowerbound}

It is also possible to consider the \emph{strongly adaptive blocking adversary}, who observes the entire state of the current round including the coin flips performed by nodes, before choosing a set of $\le \epsilon n$ nodes that are blocked from sending/receiving communication during the round (cf.\ \Cref{sec:model}).\footnote{We follow the notation of \cite{kuhn-opodis15}, where the adversary that observes the current coin flips and network configuration before making its choice for the next round is called strongly adaptive.}
Since this adversary can simulate the strongly adaptive crash-failure adversary, it is easy to see that the $\tilde\Omega(\sqrt{n})$ lower bounds shown for the latter in \cite{barjoseph} also holds for the strongly adaptive {blocking} adversary.

However, by fully exploiting the adversary's power of adaptively blocking a possibly changing set of nodes, we can strengthen this result to $\Omega(n)$ rounds in \Cref{thm:lb_fullyadaptive}.
To this end, we adapt the lower bound proof of \cite{kerenAdaptive}, who show a lower bound of $\Omega(n^2)$ on the step complexity of reaching consensus in the asynchronous shared memory model under a strongly adaptive crash-failure adversary. 
We point out that in Footnote~1 of \cite{kerenAdaptive}, the authors already observed that their lower bound construction for the asynchronous crash-failure model can also be applied to the synchronous model with mobile failures, which corresponds to our setting of having a blocking adversary:
\begin{quote}
  ``Hiding processes in a layer can be described as a round-based model with mobile failures, where a process that fails in a certain round may still take steps in further rounds [Santoro and Widmayer 1989]. The equivalence between this model and the asynchronous model is discussed by Raynal and Roy [2005].''\footnote{The references mentioned are listed as \cite{DBLP:conf/stacs/SantoroW89} and \cite{DBLP:conf/prdc/RaynalR05} in our bibliography.}
\end{quote}
In the remainder of this section, we guide the reader through the modified lower bound proof and how to apply it to the strongly adaptive blocking adversary.

Note that here we are considering binary almost-everywhere agreement rather than binary consensus and hence we need to adapt the notion of a \emph{$v$-deciding adversary} of \cite{kerenAdaptive}, for consensus input value $v \in \{0,1\}$. 
We define a \emph{configuration in round $r$} to be a tuple $C = (C(u_1),\dots,C(u_n))$, where $C(u_i)$ denotes the local state of node $u_i$ at the beginning of the round, before performing its local computation and coin flips for round $r$.
The probability distribution from which a node $u_i$ samples its random bits depends on its local configuration $C(u_i)$, and hence we denote the product probability space of the local random bits of the nodes with respect to a configuration $C$ by $X^C$.
We say that an adversary who can block $f = \epsilon n$ nodes is \emph{$v$-deciding from configuration $C$ in round $r$} if the probability to reach almost-everywhere agreement on value $v$ is more than $1 - \epsilon_r$, where we define $\epsilon_r = \frac{1}{n^{3/2}} - \frac{r}{(n-f)^3}$, similarly to Section~3 in \cite{kerenAdaptive}.
Note that the probability is computed with respect to the coin flips drawn from $X^C$ in configuration $C$. 

We say that a configuration $C$ is \emph{$v$-potent with respect to a set of adversaries $\cA$}, if there is an adversary in the set $\cA$ that is $v$ deciding from $C$ (c.f.\ \cite{mosesRajsbaum}).  
We define a configuration $C$ to be \emph{$v$-valent with respect to $\cA$} (or simply \emph{univalent}) if $C$ is $v$-potent but not $(1-v)$-potent.
Configuration $C$ is \emph{bivalent} if it is both $0$-valent and $1$-valent.
On the other hand, if $C$ is neither $0$-valent nor $1$-valent, we say that $C$ is \emph{null-valent}, and we call $C$ \emph{non-univalent} if it is either bivalent or null-valent.

The proof of \cite{kerenAdaptive} constructs a sequence of undecided configurations.
First they show (c.f.\ Lemma~4.1 in \cite{kerenAdaptive}) that there is an initial input assignment of the nodes that results in a non-univalent initial configuration, which requires at most one node to fail by crashing. 
In our setting, this requires the adversary to permanently block at most one node.

Then, in the inductive step several cases are distinguished, assuming that the algorithm has reached an undecided configuration.
To this end, a given configuration is extended by a \emph{layer}, which is a collection of steps by a subset $S$ of the nodes (at most $1$ step per each node in $S$).
We can translate such a layer from the shared memory model to a synchronous round in our message passing model by instructing the adversary to block nodes that do not take steps in the layer.
\begin{compactenum}
\item If the current configuration $C$ is null-valent, it is shown that with probability $\ge 1 - O(1/n^3)$, it is possible to extend $C$ by a layer that results in another null-valent configuration. 
(Note that in our setting, a layer simply corresponds to a single round of the algorithm.)
This can be achieved by constructing a layer that does not contain a specific set $B$ of $O(\sqrt{n\log n})$ nodes (and hence these nodes do not take steps), which is admissible in the asynchronous model. 
In our synchronous message passing model, we can achieve the same effect by instructing the adversary to block the nodes in $B$ during the current round. 
A crucial difference to the crash failure adversary of \cite{barjoseph} is that our adversary does not necessarily need to block the nodes in $B$ during future rounds and hence extending a null-valent configuration does not reduce the total amount of nodes that the adversary can block in future rounds. 
Similarly, \cite{kerenAdaptive} achieve the same effect by leveraging the asynchrony of the scheduler to prevent nodes from taking a step in a given layer. 
\item If the current configuration $C$ is bivalent, Lemma~4.3 in \cite{kerenAdaptive} shows that, by just failing at most $1$ additional node $u$, it is possible to extend $C$ to a configuration that is either bivalent, null-valent, or a so called \emph{$v$-switching configuration}, which is a $v$-valent configuration that is reached by extending a $(1-v)$-potent configuration $C$ with a $(1-v)$-deciding adversary whose strategy depends on the sampled coin flips.
  In our setting, we can achieve the same effect by permanently blocking $u$.
\item Finally, in the case where a $(1-v)$-potent configuration $C$ was extended to a configuration $C'$ that is a $v$-switching configuration, Lemma~4.4 of \cite{kerenAdaptive} shows that with probability $1 - O(1/n^{3/2})$, it is sufficient to crash at most one additional node to extend $C'$ to a configuration that is a $(1-v)$-switching configuration or non-univalent. As in the previous case, we can emulate the adversary's strategy by permanently blocking one additional node.
\end{compactenum}
The detailed proof of \Cref{thm:lb_fullyadaptive} follows along the lines of the proof of Theorem~4.8 in \cite{kerenAdaptive}, where the above inductive argument is used to obtain a sequence of $\Omega(n)$ layers (with probability $1 - o(1)$) in which the algorithm remains undecided. This corresponds to a sequence of $\Omega(n)$ rounds in our setting. (As mentioned above, \cite{kerenAdaptive} already point out the possibility of applying their construction to the synchronous model with ``mobile'' failures, which is due to the equivalence of this model with the asynchronous crash-failure model argued in \cite{DBLP:conf/prdc/RaynalR05}. Since our construction does not require new techniques, we omit the formal details of the proof.)

\begin{theorem} \label{thm:lb_fullyadaptive}
Reaching binary almost-everywhere agreement among more than half of the nodes requires at least $\Omega(n)$ rounds with probability $1 - o(1)$ against the strongly adaptive adversary, who observes the entire state of the network including coin flips performed by the nodes, if, in each round, the adversary can block a set of up to $\epsilon n$ nodes, for some constant $\epsilon>0$.
\end{theorem}

\section{Experimental Evaluation}\label{sec:experiments}

We have performed simulations of the $(k,\ell)$-majority algorithm for various network sizes and adversarial thresholds. 
As elaborated in more detail below, these results suggest that the algorithm reaches almost-everywhere agreement within $\le 2\log(n)$ rounds (on average) while being resilient against up to $n/15$ adversarial attacks per round.
Moreover, by just increasing the communication degree of each node from $6$ to $12$, i.e., executing the $(12,3)$-majority algorithm, it turns out that agreement is achievable as long as the adversary blocks no more than $n/5$ nodes per round.

\subsection{Setup}

The simulation was executed in the Erlang virtual machine and implemented in the programming language Elixir. 
We briefly describe how to simulate the round structure in this setting.\footnote{We emphasize that the goal of the simulation is to evaluate the theoretical bounds of the $(k,\ell)$-majority algorithm. In particular, we did not aim for an efficient real-world implementation of consensus.}
Each simulated node is executed as a separate (distributed) process that communicates with the other processes via message passing.
The late adversary is also modeled as a separately executing process who is informed by all participating nodes regarding their value at the start of the round, and then, at the end of the round, decides which nodes to block in the next round based on this information. 
In more detail, the adversary computes the current difference between nodes that held a $0$ respectively a $1$ (at the start of the round).
Then, the adversary sends a \texttt{blocking} respectively \texttt{unblocking} command (as a message) to each node, which informs a node about its status in the next round.
We point out that, according to the $(k,\ell)$-majority algorithm (see Section~\ref{sec:binary}), any node $u$ that is blocked discards its current value and adopts  value $\bot$ instead.
Moreover, node $u$ refrains from sending messages to other nodes during the next round and discards all received messages.
The execution starts in a balanced state w.r.t.\ the input values and is terminated if one of the following two conditions hold:
\begin{compactenum} 
  \item[(1)] the difference between the number of nodes supporting $0$ and the number of nodes supporting $1$ is $\ge (\frac{2}{3} - \epsilon)n$; 
  \item[(2)] at least $n/2$ nodes currently hold the reset value $\bot$.
\end{compactenum}
Case~(1) represents reaching almost-everywhere agreement, whereas we assume that the adversary has succeeded if Case~(2) holds and consider the trial as failed.

\subsection{Results}
For each parameter setting, we performed 1000 independent runs of the $(k,\ell)$-majority algorithm. 
The first set of results concern the $(6,3)$-majority algorithm and show that the algorithm succeeded to achieve almost-everywhere agreement in 100\% of the trials as long as $\epsilon \le 1/15$, i.e., when the adversary can block at most $n/15$ nodes in each round; see Figure~\ref{fig:sixthree}. 
The average number of rounds until termination is within $2\log(n)$ rounds, whereas the $95$th percentile is bounded by $3\log(n)$.
For $\epsilon \ge 1/14$, on the other hand, the success rate of the algorithm drops to approximately $80$\% and the 95th percentiles increase significantly. 
Since the algorithm fails when too many nodes hold the reset value $\bot$, which either occurs due to being blocked by the adversary or when a node does not receive sufficiently many messages from other nodes, we also investigated the impact of increasing the number of recipients that each node contacts to disseminate its value from $6$ to $12$.
Figure~\ref{fig:twelvethree} shows that the $(12,3)$-majority algorithm has a significantly higher resilience and achieved a.e.\ agreement in all trials for $\epsilon \le 1/5$.
It may be tempting to further increase the communication by instructing nodes to contact $24$ peers in each round, however, in our simulations this did not provide any improvements regarding the resilience threshold compared to the $(12,3)$-majority algorithm.

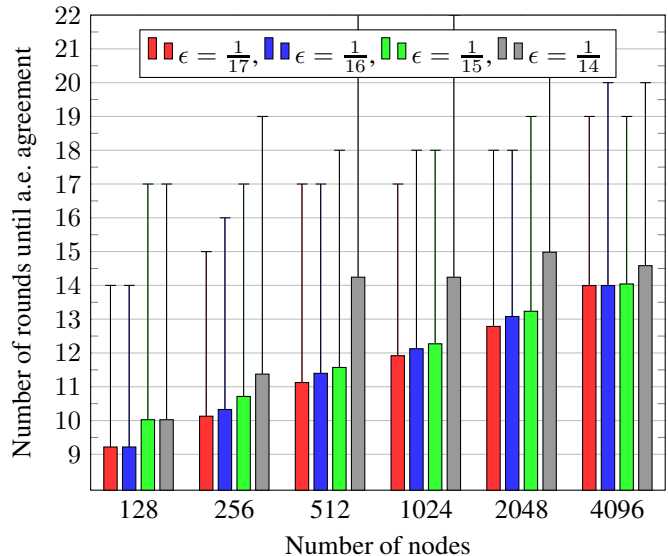
\begin{figure*}
  \centering
\begin{tikzpicture}
\begin{axis}[
  xlabel={Number of nodes},
  ylabel={Number of rounds until a.e.\ agreement},
  width=260pt,
  x tick label style={major tick length=0pt},
  xmode=log,
  log ticks with fixed point,
  xtick=data,
  xticklabels={128,256,512,1024,2048,4096},
  ytick={9,10,11,12,13,14,15,16,17,18,19,20,21,22},
  ymax=22,
  minor y tick num=1,
  ymajorgrids,
  legend style={at={(0.5,0.97)},
  legend entries={{$\epsilon=\tfrac{1}{17}$,},{$\epsilon=\tfrac{1}{16}$,},{$\epsilon=\tfrac{1}{15}$,},$\epsilon=\tfrac{1}{14}$},
  anchor=north,legend columns=-1},
  ybar,
  bar width=5pt,
]
\addplot[fill=red!80,
         error bars/.cd, 
         y dir=plus, 
         y explicit,
        ]
       plot coordinates { 
           (128,9.217) += (0,4.783)
           (256,10.128) += (0,4.872)
           (512,11.125) += (0,5.875)
           (1024,11.918) += (0,5.082)
           (2048,12.785) += (0,5.215)
           (4096,13.994) += (0,5.006)
       };
\addplot[fill=blue!80,
         error bars/.cd, 
         y dir=plus, 
         y explicit,
        ]
      plot coordinates { 
          (128,9.217) += (0,4.783)
          (256,10.327) += (0,5.673)
          (512,11.397)  += (0,5.603)
          (1024,12.125) += (0,5.875)
          (2048,13.079) += (0,4.921)
          (4096,13.994) += (0,6.006)
       };
\addplot[fill=green!80,
          error bars/.cd, 
          y dir=plus, 
          y explicit,
        ]
       plot coordinates { 
           (128,10.0243) += (0,6.9757)
           (256,10.712) += (0,6.288)
           (512,11.571) += (0,6.429)
           (1024,12.271) += (0,5.729)
           (2048,13.234) += (0,5.766)
           (4096,14.042) += (0,4.958)
       };
\addplot[fill=gray!80,
         error bars/.cd, 
         y dir=plus, 
         y explicit,
        ]
       plot coordinates { 
           (128,10.0243)  += (0,17-10.0243) 
           (256,11.374)   += (0,19-11.374)
           (512,14.243)   += (0,27-11.374)
           (1024,14.2417) += (0,22-14.2417)
           (2048,14.9819) += (0,21-14.9819)
           (4096,14.5827) += (0,20-14.5827)
       };
\end{axis}
\end{tikzpicture}
\caption{\small Running time of the $(6,3)$-Majority Algorithm when the adversary can block $\epsilon n$ nodes, for $\epsilon \in \{\tfrac{1}{17},\tfrac{1}{16},\tfrac{1}{15},\tfrac{1}{14}\}$. The filled bars represent the average number of rounds and the error-bars show the 95th percentiles of the 1000 trials for each parameter setting. Note that the algorithm succeeded in all trials for $\epsilon \in \{\tfrac{1}{17},\tfrac{1}{16},\tfrac{1}{15}\}$. For the case $\epsilon=\tfrac{1}{14}$ (shown as the grey bar), the success rate was reduced to $81$\% for $n=4096$, whereas the running time increased for smaller $n$; the (clipped) 95th-percentiles for $n=\{512,1024\}$ and $\epsilon=1/14$ are at $27$ and $22$ rounds, respectively. Almost all trials failed for larger values of $\epsilon$.}
\label{fig:sixthree}
\end{figure*}

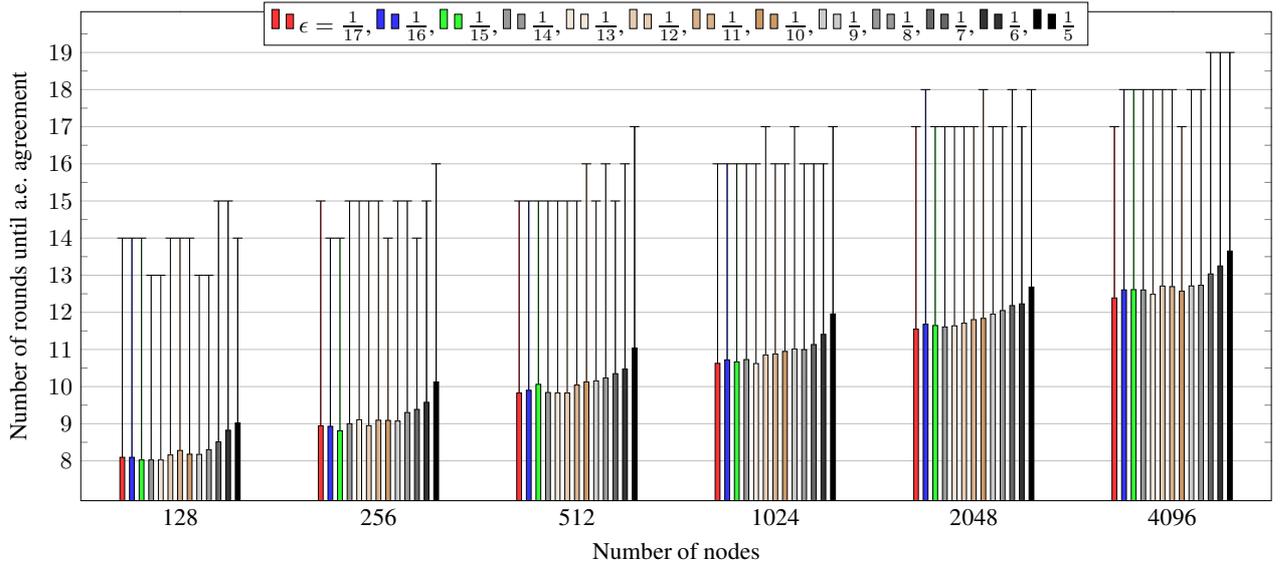
\begin{figure*}
  \centering
\begin{tikzpicture}[scale=.9,transform shape]
\begin{axis}[
  xlabel={Number of nodes},
  ylabel={Number of rounds until a.e.\ agreement},
  width=540,
  height=250,
  scaled ticks=false,
  x tick label style={major tick length=0pt},
  log basis x=2,
  xmode=log,
  log ticks with fixed point,
  ytick={8,9,10,11,12,13,14,15,16,17,18,19},
  xtick=data,
  xticklabels={128,256,512,1024,2048,4096},
  minor y tick num=1,
  ymajorgrids,
  legend style={at={(0.5,1.02)},
  legend entries={ {$\epsilon=\tfrac{1}{17}$,},{$\tfrac{1}{16}$,},{$\tfrac{1}{15}$,},{$\tfrac{1}{14}$,},{$\tfrac{1}{13}$,},{$\tfrac{1}{12}$,},{$\tfrac{1}{11}$,},{$\tfrac{1}{10}$,},{$\tfrac{1}{9}$,},{$\tfrac{1}{8}$,},{$\tfrac{1}{7}$,},{$\tfrac{1}{6}$,},{$\tfrac{1}{5}$}},
  anchor=north,legend columns=-1},
  ybar,
  bar width=2pt,
]
\addplot[fill=red!80,
         error bars/.cd, 
         y dir=plus, 
         y explicit,
        ]
       plot coordinates { 
           (128,8.095) += (0,14-8.095)
           (256,8.945) += (0,15-8.945)
           (512,9.828) += (0,15-9.828)
           (1024,10.626) += (0,16-10.626)
           (2048,11.548) += (0,17-11.548)
           (4096,12.387) += (0,17-12.387)
       };
\addplot[fill=blue!80,
         error bars/.cd, 
         y dir=plus, 
         y explicit,
        ]
      plot coordinates { 
           (128,8.095) += (0,14-8.095)
           (256,8.93) += (0,14-8.93)
           (512,9.901) += (0,15-9.901)
           (1024,10.72) += (0,16-10.72)
           (2048,11.68) += (0,18-11.68)
           (4096,12.605) += (0,18-12.605)
       };
\addplot[fill=green!80,
          error bars/.cd, 
          y dir=plus, 
          y explicit,
        ]
       plot coordinates { 
           (128,8.024) += (0,14-8.024)
           (256,8.808) += (0,14-8.808)
           (512,10.061) += (0,15-10.061)
           (1024,10.666) += (0,16-10.666)
           (2048,11.646) += (0,17-11.646)
           (4096,12.612) += (0,18-12.612)
       };
\addplot[fill=gray!80,
         error bars/.cd, 
         y dir=plus, 
         y explicit,
        ]
       plot coordinates { 
           (128,8.024) += (0,13-8.024)
           (256,9) += (0,15-9)
           (512,9.841) += (0,15-9.841)
           (1024,10.728) += (0,16-10.728)
           (2048,11.607) += (0,17-11.607)
           (4096,12.603) += (0,18-12.603)
       };
\addplot[fill=brown!20,
         error bars/.cd, 
         y dir=plus, 
         y explicit,
        ]
       plot coordinates { 
           (128,8.024) += (0,13-8.024)
           (256,9.107) += (0,15-9.107)
           (512,9.827) += (0,15-9.827)
           (1024,10.621) += (0,16-10.621)
           (2048,11.638) += (0,17-11.638)
           (4096,12.486) += (0,18-12.486)
       };
 \addplot[fill=brown!40,
          error bars/.cd, 
          y dir=plus, 
          y explicit,
         ]
        plot coordinates { 
            (128,8.159) += (0,14-8.159)
            (256,8.948) += (0,15-8.948)
            (512,9.827) += (0,15-9.827)
            (1024,10.851) += (0,17-10.851)
            (2048,11.707) += (0,17-11.707)
            (4096,12.708) += (0,18-12.708)
        };
 \addplot[fill=brown!60,
          error bars/.cd, 
          y dir=plus, 
          y explicit,
         ]
        plot coordinates { 
            (128,8.28) += (0,14-8.28)
            (256,9.095) += (0,15-9.095)
            (512,10.041) += (0,15-10.041)
            (1024,10.878) += (0,16-10.878)
            (2048,11.804) += (0,17-11.804)
            (4096,12.693) += (0,18-12.693)
        };
 \addplot[fill=brown!80,
          error bars/.cd, 
          y dir=plus, 
          y explicit,
         ]
        plot coordinates { 
            (128,8.18) += (0,14-8.18)
            (256,9.09) += (0,14-9.09)
            (512,10.126) += (0,16-10.126)
            (1024,10.949) += (0,16-10.949)
            (2048,11.838) += (0,18-11.838)
            (4096,12.572) += (0,17-12.572)
        };
        
 \addplot[fill=black!20,
          error bars/.cd, 
          y dir=plus, 
          y explicit,
         ]
        plot coordinates { 
            (128,8.175) += (0,13-8.175)
            (256,9.073) += (0,15-9.073)
            (512,10.152) += (0,15-10.152)
            (1024,11.012) += (0,17-11.012)
            (2048,11.95) += (0,17-11.95)
            (4096,12.711) += (0,18-12.711)
        };
 \addplot[fill=black!40,
          error bars/.cd, 
          y dir=plus, 
          y explicit,
         ]
        plot coordinates { 
            (128,8.301) += (0,13-8.301)
            (256,9.302) += (0,15-9.302)
            (512,10.233) += (0,16-10.233)
            (1024,10.996) += (0,16-10.996)
            (2048,12.049) += (0,17-12.049)
            (4096,12.729) += (0,18-12.729)
        };
 \addplot[fill=black!60,
          error bars/.cd, 
          y dir=plus, 
          y explicit,
         ]
        plot coordinates { 
            (128,8.511) += (0,15-8.511)
            (256,9.386) += (0,14-9.386)
            (512,10.348) += (0,15-10.348)
            (1024,11.133) += (0,16-11.133)
            (2048,12.182) += (0,18-12.182)
            (4096,13.034) += (0,19-13.034)
        };
 \addplot[fill=black!80,
          error bars/.cd, 
          y dir=plus, 
          y explicit,
         ]
        plot coordinates { 
            (128,8.825) += (0,15-8.825)
            (256,9.577) += (0,15-9.577)
            (512,10.472) += (0,16-10.472)
            (1024,11.407) += (0,16-11.407)
            (2048,12.228) += (0,17-12.228)
            (4096,13.249) += (0,19-13.249)
        };
 \addplot[fill=black,
          error bars/.cd, 
          y dir=plus, 
          y explicit,
         ]
        plot coordinates { 
            (128,9.02) += (0,14-9.02)
            (256,10.125) += (0,16-10.125)
            (512,11.041) += (0,17-11.041)
            (1024,11.956) += (0,17-11.956)
            (2048,12.681) += (0,18-12.681)
            (4096,13.651) += (0,19-13.651)
        };
\end{axis}
\end{tikzpicture}
\caption{\small Running time of the $(12,3)$-Majority Algorithm when the adversary can block $\epsilon n$ nodes, for $\epsilon \in \{\tfrac{1}{17},\tfrac{1}{16},\dots,\tfrac{1}{5}\}$. The filled bars represent the average number of rounds and the error bars show the 95th percentiles. The algorithm achieved a.e.\ agreement in all of the $1000$ trials that were executed for each choice of $\epsilon$. When $\epsilon\ge \tfrac{1}{4}$, the success rate deteriorates significantly and the algorithm achieves a.e.\ agreement in fewer than $1\%$ of the trials.}
\label{fig:twelvethree}
\end{figure*}

\section{Conclusion and Future Work}
We have initiated the study of reaching agreement in the late adversarial model, where the adversary has a view that is outdated by $1$ round.
Our results raise the question if increasing the ``lateness'' of the adversary to a larger number of rounds results in improvement of the time complexity of our algorithms.
Our lower bound gives a negative answer to this question.
It would be interesting to investigate the impact of lateness on the complexity of distributed algorithms also for other problems, since it is realistic to assume that typically an adversary will not have the most up-to-date information about a distributed system.

\bibliographystyle{plain}
\bibliography{references}

\end{document}